\documentclass[acmsmall,nonacm]{acmart}

\usepackage{amsthm, amsmath, mathtools}
\usepackage{upgreek}
\usepackage{amsfonts}
\usepackage{enumitem}
\usepackage[normalem]{ulem}
\usepackage[linesnumbered,vlined,ruled,commentsnumbered]{algorithm2e}
\usepackage[para,online,flushleft]{threeparttable}
\usepackage{multirow}
\usepackage{multicol}
\usepackage{hyperref}
\usepackage{url}
\usepackage{bm}
\usepackage[labelfont=bf,skip=0pt]{caption,subfig}
\usepackage{flushend}

\usepackage{tabularx}

\usepackage[english]{babel}
\usepackage{graphicx}
\usepackage{thmtools}

\usepackage{listings}
\usepackage{tablefootnote}
\usepackage{enumitem}
\usepackage{bm}
\usepackage{comment}
\usepackage{tikz}

\DeclareMathOperator*{\argmin}{arg\,min}

\newcommand{\opt}{\textsf{opt}}

\newcommand{\Q}{\boldsymbol{q}}
\newcommand{\E}{\mathcal{E}}

\newcommand{\T}{\mathcal{T}}

\newcommand{\eps}{\varepsilon}
\renewcommand{\epsilon}{\varepsilon}
\renewcommand{\Re}{\mathbb{R}}

\renewcommand{\paragraph}[1]{\smallskip\noindent{\bf {#1. }}}

\newcommand{\allattr}{{\mathbf A}}

\newcommand{\allrel}{{\mathbf R}}

\newcommand{\attr}{{\texttt{attr}}}

\newcommand{\dom}{{\texttt{dom}}}

\newcommand{\fhw}{\mathsf{fhw}}

\newcommand{\query}{\Q}
\newcommand{\D}{\mathbf{D}}
\newcommand{\I}{\D}

\newcommand{\dist}{\phi}

\newcommand{\kmedian}{\textbf{v}}
\newcommand{\kmeans}{\mu}
\renewcommand{\opt}{\mathsf{OPT}}
\newcommand{\optD}{\mathsf{OPT}_{\textsf{disc}}}
\newcommand{\kapprox}{X}
\newcommand{\coreset}{\mathcal{C}}
\newcommand{\ret}{\mathcal{S}}
\newcommand{\diameter}{\textsf{diam}}
\newcommand{\tree}{\mathcal{T}}
\newcommand{\node}{u}
\newcommand{\all}{\Q_\node(\I)}
\newcommand{\error}{\mathcal{E}}
\newcommand{\timeMedian}{T^{\mathsf{med}}}
\newcommand{\timeMeans}{T^{\mathsf{mean}}}
\newcommand{\rootnode}{\rho}
\renewcommand{\O}{\tilde{O}}
\newcommand{\kmedianAlg}{\mathsf{GkMedianAlg}}
\newcommand{\DkmedianAlg}{\mathsf{DkMedianAlg}}
\newcommand{\kmeansAlg}{\mathsf{GkMeansAlg}}
\newcommand{\DkmeansAlg}{\mathsf{DkMeansAlg}}
\newcommand{\nn}{\mathcal{N}}

\def\mparagraph#1{\par\noindent\textbf{#1.}\quad}

\allowdisplaybreaks
\begin{document}

\title{Improved Approximation Algorithms for Relational Clustering} 

\author{Aryan Esmailpour}
\affiliation{%
  \institution{Department of Computer Science, University of Illinois Chicago}
  \city{Chicago}
  \country{USA}}
\email{aesmai2@uic.edu}

\author{Stavros Sintos}
\affiliation{%
  \institution{Department of Computer Science, University of Illinois Chicago}
  \city{Chicago}
  \country{USA}}
\email{stavros@uic.edu}

\begin{abstract}
   Clustering plays a crucial role in computer science, facilitating data analysis and problem-solving across numerous fields. By partitioning large datasets into meaningful groups, clustering reveals hidden structures and relationships within the data, aiding tasks such as unsupervised learning, classification, anomaly detection, and recommendation systems. Particularly in relational databases, where data is distributed across multiple tables, efficient clustering is essential yet challenging due to the computational complexity of joining tables. This paper addresses this challenge by introducing efficient algorithms for $k$-median and $k$-means clustering on relational data without the need for pre-computing the join query results. For the relational $k$-median clustering, we propose the first efficient relative approximation algorithm.
   For the relational $k$-means clustering, our algorithm significantly improves both the approximation factor and the running time of the known relational $k$-means clustering algorithms, which suffer either from large constant approximation factors, or expensive running time.
   Given a join query $\Q$ and a database instance $\I$ of $O(N)$ tuples, for both $k$-median and $k$-means clustering on the results of $\Q$ on $\I$, we propose randomized $(1+\eps)\gamma$-approximation algorithms that run in roughly $O(k^2N^{\fhw})+T_\gamma(k^2)$ time, where $\eps\in (0,1)$ is a constant parameter decided by the user, $\fhw$ is the fractional hyper-tree width of $\Q$, while $\gamma$ and $T_\gamma(x)$ are respectively the approximation factor and the running time of a traditional clustering algorithm in the standard computational setting over $x$ points.
   
\end{abstract}




\maketitle
\section{Introduction}
Clustering is a fundamental process in computer science, serving as a vital tool for data analysis, pattern recognition, and problem-solving across various domains. Through clustering, large datasets are partitioned into meaningful groups, unveiling hidden structures and relationships within the data. 
In machine learning, clustering algorithms are used to enable tasks such as classification, anomaly detection, and recommendation systems.
Its significance lies in its ability to transform raw data into useful knowledge, empowering researchers, and businesses to make informed decisions.


In relational databases, data is gathered and stored across various tables.
Each table consists of a set of tuples and two tuples stored in different tables might refer to the same entity.
Relational data is decoupled into different relational tables, and we cannot obtain the full data unless all the tables are joined.
This setting is quite common in real database management systems (DBMS). 
As shown in the DB-engines study~\cite{link1} the vast majority of database systems are relational DBMS.
Kaggle surveys~\cite{kaggle} demonstrate that most of the learning tasks encountered by data scientists involve relational data. More specifically, 70\% of database systems are relational DBMS and 65\% of the data sets in learning tasks are relational data.
As mentioned in~\cite{chen2022coresets}, relational data is expected to reach \$122.38 billion by 2027~\cite{link:relational} in investments.

In order to explore and process relational data, usually two steps are required: data preparation, and data processing. In the data preparation step, tuples from different tables are joined to construct useful data, while in data processing, an algorithm (for example a clustering algorithm) is performed on the join results to analyze the data. This two-step approach is usually too expensive because the size of the join results can be polynomially larger than the total size of the input tables~\cite{ngo2018worst, ngo2014skew}.

In this paper, given a (full) conjunctive query and a database instance, we study efficient $k$-median and $k$-means clustering on the results of the query without first computing the query results.
While there are recent papers~\cite{chen2022coresets, curtin2020rk, moseley2021relational} on designing clustering algorithms on relational data, they usually suffer from i) large additive approximation error, ii) large constant relative approximation error, and iii) expensive running time.

Despite recent attention to relational clustering, the problem of efficiently solving relational $k$-median clustering without additive approximation error remained unsolved.
In this paper, we propose the first efficient relative approximation algorithm for the $k$-median clustering on relational data. Furthermore, by extending our methods, we design an algorithm for the $k$-means clustering on relational data that dominates (with respect to both the running time and the approximation ratio) all the known relative approximation algorithms.

\subsection{Notation and problem definition}
\noindent{\bf Conjunctive Queries.}
We are given a database schema $\allrel$ over a set of $d$ attributes $\allattr$. The database schema $\allrel$ contains $m$ relations $R_1, \ldots, R_m$. Let $\allattr_j\subseteq \allattr$ be the set of attributes associated with relation $R_j\in \allrel$. For an attribute $A\in \allattr$, let $\dom(A)$ be the domain of attribute $A$. We assume that $\dom(A)=\Re$ for every $A\in \allattr$. Let $\I$ be a database instance over the database schema $\allrel$. For simplicity, we assume that each relation $R_j$ contains $N$ tuples in $\I$.
Throughout the paper, we consider data complexity i.e., $m$ and $d$ are constants, while $N$ is a large integer.
We use $R_j$ to denote both the relation and the set of tuples from $\I$ stored in the relation.
For a subset of attributes $B\subseteq \allattr$ and a set $Y\subset\Re^d$, let $\pi_{B}(Y)$ be the set containing the projection of the tuples in $Y$ onto the attributes $B$. Notice that two different tuples in $Y$ might have the same projection on $B$, however, $\pi_{B}(Y)$ is defined as a set, so the projected tuple is stored once. We also define the multi-set $\widebar{\pi}_{B}(Y)$ so that if for two tuples $t_1, t_2\in Y$ it holds that $\pi_{B}(t_1)=\pi_B(t_2)$, then the tuple $\pi_B(t_1)$ exists more than once in $\widebar{\pi}_B(Y)$.

Following the related work on relational clustering~\cite{chen2022coresets, curtin2020rk, moseley2021relational}, we are given a full conjunctive query (join query) $\Q:=R_1\Join \ldots \Join R_m$. The set of results of a join query $\Q$ over the database instance $\I$ is defined as $\Q(\I)=\{t\in \Re^d\mid \forall j\in[1,m]:\pi_{\allattr_j}(t)\in R_j\}$.
For simplicity, all our algorithms are presented assuming that $\Q$ is an acyclic join query, however in the end we extend to any general join query. A join query $\Q$ is acyclic if there exists a tree, called join tree, such that the nodes of the tree are the relations in $\allrel$ and for every attribute $A\in \allattr$, the set of nodes/relations that contain $A$ form a connected component.
Let $\rho^*(\Q)$ be the fractional edge cover of query $\Q$, which is a parameter that bounds the number of join results $\Q(\I)$ over any database instance. More formally, for every database instance $\I'$ with $O(N)$ tuples in each relation, it holds that 
$|\Q(\I')|=O(N^{\rho^*(\Q)})$ as shown in~\cite{atserias2013size}.

We use the notation $\fhw(\Q)$ to denote the \emph{fractional hypertree width}~\cite{gottlob2014treewidth} of the query $\Q$.
The fractional hypertree width roughly measures how close $\Q$ is to being acyclic.
For every acyclic join query $\Q$, we have $\fhw(\Q)=1$.
Given a cyclic join query $\Q$, we convert it to an equivalent acyclic query such that each relation is the result of a (possibly cyclic) join query with fractional edge cover at most $\fhw(\Q)$. Hence, a cyclic join query over a database instance with $O(N)$ tuples per relation can be converted, in $O(N^{\fhw(\Q)})$ time, to an equivalent acyclic join query over a database instance with $O(N^{\fhw(\Q)})$ tuples per relation~\cite{atserias2013size}.
A more formal definition of $\fhw$ is given in Appendix~\ref{appndx:generalQueries}. If $\Q$ is clear from the context, we write $\fhw$ instead of $\fhw(\Q)$.

For two tuples/points\footnote{The terms points and tuples are used interchangeably.} $p,q\in \Re^d$ let $\dist(p,q)=||p-q||=\left(\sum_{j=1,\ldots, d}(\pi_{A_j}(p)-\pi_{A_j}(q))^2\right)^{1/2}$ be the Euclidean distance between $p$ and $q$.
Throughout the paper, we use the Euclidean distance to measure the error of the clustering.

\paragraph{Clustering}
In this paper, we focus on $k$-median and $k$-means clustering. We start with some useful general definitions.
Let $P$ be a set of points in $\Re^d$ and let $C$
be a set of $k$ centers/points in $\Re^d$.
Let $w:\Re^d\rightarrow \Re_{>0}$ be a weight function such that $w(p)$ is the weight of point $p\in P$.
For a point $p\in\Re^d$, let $\dist(p,C)=\min_{c\in C}\dist(p,c)$. We define
\vspace{-0.1em}
$$\kmedian_{C}(P)=\sum_{p\in P}w(p)\dist(p,C),\quad \text{ and } \quad \kmeans_{C}(P)=\sum_{p\in P}w(p)\dist^2(p,C).$$
If $P$ is an unweighted set, then $w(p)=1$ for every $p\in P$.

\hspace{-1em}\textbf{$k$-median clustering}: Given a weight function $w$, a set of points $P$ in $\Re^d$ and a parameter $k$, the goal is to find a set $C\subset \Re^d$ with $|C|=k$ such that $\kmedian_{C}(P)$ is minimized. This is also called the \emph{geometric $k$-median clustering} problem. Equivalently, we define the \emph{discrete $k$-median clustering} problem, where the goal is to find a set $C\subseteq P$ with $|C|=k$ such that $\kmedian_{C}(P)$ is minimized.
Let $\opt(P)=\argmin_{S\in \Re^d, |S|= k}\kmedian_S(P)$ be a set of $k$ centers in $\Re^d$ with the minimum $\kmedian_{\opt(P)}(P)$. For the discrete $k$-median problem let $\optD(P)=\argmin_{S\subseteq P, |S|= k}\kmedian_S(P)$. It is always true that $\kmedian_{\opt(P)}(P)\leq \kmedian_{\optD(P)}(P)\leq 2\kmedian_{\opt(P)}(P)$.
Let $\kmedianAlg_\gamma$  (resp. $\DkmedianAlg_\gamma$) be a (known) $\gamma$-approximation algorithm for the geometric $k$-median problem (resp. discrete $k$-median problem) in the standard computational setting\footnote{In the standard computational setting data is stored in one table.} that runs in $\timeMedian_\gamma(|P|)$ time, where $\gamma$ is a constant.

\hspace{-1em}\textbf{$k$-means clustering}: Given a weight function $w$, a set of points $P$ in $\Re^d$ and a parameter $k$, the goal is to find a set $C\subset \Re^d$ with $|C|=k$ such that $\kmeans_{C}(P)$ is minimized. This is also called the \emph{geometric $k$-means clustering} problem. Equivalently, we define the \emph{discrete $k$-means clustering} problem, where the goal is to find a set $C\subseteq P$ with $|C|=k$ such that $\kmeans_{C}(P)$ is minimized.
Let $\opt(P)=\argmin_{S\in \Re^d, |S|= k}\kmeans_S(P)$ be a set of $k$ centers in $\Re^d$ with the minimum $\kmeans_{\opt(P)}(P)$. For the discrete $k$-means problem let $\optD(P)=\argmin_{S\subseteq P, |S|= k}\kmeans_S(P)$. It is always true that $\kmeans_{\opt(P)}(P)\leq \kmeans_{\optD(P)}(P)\leq 4\kmeans_{\opt(P)}(P)$.
Let $\kmeansAlg_\gamma$ (resp. $\DkmeansAlg_\gamma$) be a (known) $\gamma$-approximation algorithm for the geometric $k$-means problem (resp. discrete $k$-means problem) in the standard computational setting that runs in $\timeMeans_\gamma(|P|)$ time, where $\gamma$ is a constant.


For simplicity, we use the same notation $\timeMedian_\gamma(\cdot)$ for the running time of $\kmedianAlg$ and $\DkmedianAlg$. Similarly, we use the same notation $\timeMeans_\gamma(\cdot)$ for the running time of $\kmeansAlg$ and $\DkmeansAlg$. We also use the same notation $\opt(\cdot)$ for the optimum solution for $k$-means and $k$-median clustering problems. It is always clear from the context whether we are referring to $k$-means or $k$-median clustering.

In this paper, we study (both geometric and discrete) $k$-median and $k$-means clustering on the result of a join query:

\begin{definition}
\label{def:kmedian}
    [Relational $k$-median clustering] Given a database instance $\I$, a join query $\Q$, and a positive integer parameter $k$, the goal is to find a set $\ret\subseteq \Re^d$ of size $|\ret|=k$ such that $\kmedian_C(\Q(\I))$ is minimized. In the discrete relational $k$-median clustering the set $\ret$ should be a subset of $\Q(\I)$.
\end{definition}

\begin{definition}
\label{def:kmeans}
    [Relational $k$-means clustering] Given a database instance $\I$, a join query $\Q$, and a positive integer parameter $k$, the goal is to find a set $\ret\subseteq \Re^d$ of size $|\ret|=k$ such that the $\kmeans_\ret(\Q(\I))$ is minimized. In the discrete relational $k$-means clustering the set $\ret$ should be a subset of $\Q(\I)$.
\end{definition}

By relational $k$-median or $k$-means clustering, we are referring to the geometric versions, unless explicitly stated otherwise for the discrete variants. We propose results for both versions.
We note that the clustering problem on relational data is defined over the unweighted set $\Q(\I)$.
In order to propose efficient algorithms we will need to construct weighted subset of tuples (\emph{coresets}) that are used to derive small approximation factors for the relational clustering problems.

\paragraph{Approximation}
We say that an algorithm is a relative $\beta$-approximation algorithm for the relational $k$-median clustering if it returns a set $\ret$ of size $k$ such that $\kmedian_\ret(\Q(\I))\leq \beta \cdot\kmedian_{\opt(\Q(\I))}\Q(\I)$. Next, we say that an algorithm is an additive $\beta$-approximation algorithm for the relational $k$-median clustering if it returns a set $\ret$ of size $k$ such that $\kmedian_\ret(\Q(\I))\leq \kmedian_{\opt(\Q(\I))}\Q(\I)+\beta$. 
In this paper, we focus on relative approximation algorithms, so when we refer to a $\beta$-approximation algorithm we always mean relative $\beta$-approximation.
Equivalently, we define relative and additive approximation algorithms for the relational $k$-means clustering.

\vspace{-0.5em}
\subsection{Related work}
There is a lot of research on $k$-median and $k$-means clustering in the standard computational setting. For the $k$-median clustering there are several polynomial time algorithms with constant approximation ratio~\cite{li2013approximating, charikar1999constant, arya2001local}.
For the $k$-means clustering, there are also several constant approximation algorithms such as~\cite{kanungo2002local}. In practice, a local search algorithm~\cite{lloyd1982least} is mostly used with a $O(\log k)$ approximation ratio. 
Furthermore, coresets (formal definition in Section~\ref{sec:prelim}) have been used to design efficient clustering algorithms~\cite{har2004coresets, bachem2018scalable, har2005smaller}.
Coresets are also used to propose efficient clustering algorithms in the streaming setting~\cite{braverman2017clustering, guha2003clustering} or the MPC model~\cite{bahmani2012scalable, ene2011fast}.
All these algorithms work in the standard computational setting and it is not clear how to efficiently extend them to relational clustering.

In its most general form, relational clustering is referring to clustering objects connected by links representing persistent relationships between them.
Different variations of relational clustering have been studied over the years in the database and data mining community, for example~\cite{long2010relational, frigui2007clustering, havens2010clustering, anderson2008clustering, breiger1975algorithm, kirsten1998relational, neville2003clustering}. However, papers in this line of work either do not handle clustering on join results or their methods do not have theoretical guarantees.

The discrete relational $k$-means clustering problem has been recently studied in the literature.
Khamis et al.~\cite{khamis2020functional}, gave an efficient implementation of the Lloyd’s $k$-means heuristic in the relational setting, however it is known that the algorithm terminates in a local minimum without any guarantee on the approximation factor.
Relative approximation algorithms are also known for the relational $k$-means clustering problem.
Curtin et al.~\cite{curtin2020rk} construct a weighted set of tuples (\emph{grid-coreset}) such that an approximation algorithm for the weighted $k$-means clustering in the grid-coreset returns an approximation solution to the relational $k$-means clustering. Their algorithm runs in $\O(k^m N^\fhw +\timeMeans_\gamma(k^m))$ time and has a $(\gamma^2+4\gamma\sqrt{\gamma}+4\gamma)$-approximation factor. For some join queries, the approximation factor can be improved to
$(4\gamma+2\sqrt{\gamma}+1)$.
Moseley et al.~\cite{moseley2021relational} designed a relational implementation of the $k$-means++ algorithm~\cite{aggarwal2009adaptive, arthur2007k} to derive a better weighted coreset. For a constant $\eps\in(0,1)$, their algorithm runs in 
$O(k^4N^\fhw\log N +k^2N^\fhw\log^9 N+\timeMeans_\gamma(k\log N))$ expected
time and has a $(320+644(1+\eps)\gamma)$-approximation factor.
No efficient relative approximation algorithm is known for the relational $k$-median problem.

Additive approximation algorithms are also known for relational clustering problems.
Chen et al.~\cite{chen2022coresets}, constructed coresets for empirical risk minimization problems in relational data. Their algorithm is quite general and can support relational clustering, i.e., $k$-median and $k$-means can be formulated as risk minimization problems. They work independently in every relation to compute a good enough coreset and then by the \emph{aggregation tree} algorithm they merge the solutions carefully using properties of the $k$-center clustering. We note that the authors do not design algorithms for relational clustering problems, instead, they compute a coreset such that any approximation algorithm (for $k$-median or $k$-means) on the coreset returns an (additive) approximation of the relational $k$-median or $k$-means problem. Running 
$\kmedianAlg$ on top of their coreset leads to a randomized algorithm that runs in 
$O(N^\fhw+ \timeMedian_\gamma(1))$
time and has an $\eps\cdot\gamma\cdot\diameter(\Q(\I))$ additive approximation term, with high probability, where $\diameter(\Q(\I))$ is the largest Euclidean distance between two tuples in $\Q(\I)$.
For the relational $k$-means clustering, 
their algorithm runs in
$O(N^\fhw+\timeMeans_\gamma(1))$ 
time, and has an $\eps\cdot\gamma\cdot\diameter^2(\Q(\I))$ additive approximation term, with high probability. In all cases we assume that $\eps\in (0,1)$ is a small constant.

Finally, there is a lot of recent work on relational algorithms for learning problems such as, linear regression and factorization~\cite{rendle2013scaling, khamis2018ac, kumar2015learning, schleich2016learning}, SVMs~\cite{abo2021relational, abo2018database, yang2020towards}, Independent Gaussian Mixture models~\cite{cheng2019nonlinear, cheng2021efficient}. In~\cite{schleich2019learning} the authors give a nice survey about learning over relational data.
Generally, in databases there is an interesting line of work solving combinatorial problems over relational data without first computing the join results, such as ranked enumeration~\cite{deep2021ranked, deep2022ranked, tziavelis2020optimal},  quantiles~\cite{tziavelis2023efficient}, direct access~\cite{carmeli2023tractable}, diversity~\cite{merkl2023diversity, arenas2024towards, agarwal2025computing}, and top-$k$~\cite{tziavelis2020optimal}.

\vspace{-0.5em}
\subsection{Our results}
\begin{table*}[t]
\begin{tabular}{|c|c|c|c|c|}
\hline
\textbf{Problem}&\textbf{Method}&\textbf{Approximation}&\textbf{Running time}&\textbf{Type}\\\hline
\multirow{2}{3.9em}{$k$-median}&\textbf{NEW}&$(2+\eps)\gamma$&$\O(k^{2d+2}N^\fhw+\timeMedian_\gamma(k^2))$&D\\\cline{2-5}
&\textbf{NEW}&$(2+\eps)\gamma$&$\O(k^2N^\fhw+k^4+\timeMedian_\gamma(k^2))$&R\\\hline\hline
\multirow{4}{3.9em}{$k$-means}
&\cite{curtin2020rk}&$\gamma^2+4\gamma\sqrt{\gamma}+4\gamma$&$\O(k^m N^\fhw +\timeMeans_\gamma(k^m))$&D\\\cline{2-5}
&\cite{moseley2021relational}&$320+644(1+\eps)\gamma$&$\O(k^4N^\fhw+\timeMeans_\gamma(k))$&R\\\cline{2-5}
&\textbf{NEW}&$(4+\eps)\gamma$&$\O(k^{2d+2}N^\fhw+\timeMeans_\gamma(k^2))$&D\\\cline{2-5}
&\textbf{NEW}&$(4+\eps)\gamma$&$\O(k^2N^
\fhw+k^4+\timeMeans_\gamma(k^2))$&R\\\hline
\end{tabular}
\caption{Comparison of our new algorithms with the state-of-the-art relative approximation algorithms. For our new algorithms we show the approximation for the discrete relational $k$-median and $k$-means clustering problems. For the geometric version of the studied problems the approximation factor of our algorithms is always $(1+\eps)\gamma$.
The running time is shown in data complexity. We assume that $\eps\in(0,1)$ is a small constant. The notation $\O$ is used to hide $\log^{O(1)}N$ factors from the running time. $\timeMedian_\gamma(y)$ (resp. $\timeMeans_\gamma(y)$) is the running time of a known $\gamma$-approximation algorithm over $y$ points for the $k$-median (resp. $k$-means) clustering in the standard computational setting. $\fhw$ is the fractional hypertree width of $\Q$.
The number of attributes in the query $\Q$ is denoted by $d$. The letter R stands for randomized, while D stands for deterministic algorithm.
}
\label{table:results}
\vspace{-1em}
\end{table*}
The main results for the discrete relational clustering in this paper are summarized in Table~\ref{table:results}. In all cases, we assume that $\eps\in(0,1)$ is a small constant decided by the user.

First, we present a deterministic $(1+\eps)\gamma$-approximation algorithm for the (geometric) relational $k$-median clustering that runs in $O(k^{2d+2}N^\fhw\log^{d+2} N + \timeMedian_\gamma(k^2\log N))$ time.
Then, we show a randomized $(1+\eps)\gamma$-approximation algorithm for the (geometric) relational $k$-median clustering that works with high probability and runs in $O(k^2N^\fhw\log N + k^4\log^3 (N)\log^d (k)+\timeMedian_\gamma(k^2\log N))$ time. If $k^2\log (N)<\log^{d}(k)$ the running time can be improved to $O(k^2N^\fhw\log N + k^6\log^4 (N)+\timeMedian_\gamma(k^2\log N))$.
In the discrete case, the approximation factor is $(2+\eps)\gamma$. These are the first known efficient relative approximation algorithms for the relational $k$-median clustering problem.

We extend our methods to show algorithms with the same guarantees for the relational $k$-means clustering. We give a deterministic $(1+\eps)\gamma$-approximation algorithm for the (geometric) relational $k$-means clustering that runs in $O(k^{2d+2}N^\fhw\log^{d+2} N + \timeMeans_\gamma(k^2\log N))$ time. Furthermore, we give a randomized $(1+\eps)\gamma$-approximation algorithm  that runs in $O(k^2N^\fhw\log N + k^4\log^3 (N)\log^d (k)+\timeMeans_\gamma(k^2\log N))$ time. If $k^2\log (N)<\log^{d}(k)$ the running time can be improved to $O(k^2N^\fhw\log N + k^6\log^4 (N)+\timeMeans_\gamma(k^2\log N))$. In the discrete case the approximation factor is $(4+\eps)\gamma$. Our algorithms significantly improve both the approximation factor and the running time of the known algorithms on relational $k$-means clustering.
Due to space limit, in the next sections we focus on the relational $k$-median clustering. We show all the results for the $k$-means clustering in Appendix~\ref{appndx:kmeans}.


\mparagraph{Remark 1} While our algorithms work for both acyclic and cyclic join queries, we first present all our results assuming that $\Q$ is an acyclic join query. In Section~\ref{sec:cycRes} we extend our results for every join query using the generalized hypetree decomposition~\cite{gottlob2014treewidth}.
From now on we consider that $\Q$ is an acyclic join query, so $\fhw=1$.

\mparagraph{Remark 2} For the relational $k$-means clustering, the algorithm in~\cite{moseley2021relational} is generally faster than the other algorithms for $k=O(1)$, by $\log N$ factors. 

\section{Preliminaries}
\label{sec:prelim}
\subsection{Coreset}
We give the definition of coresets for both the relational $k$-median and $k$-means clustering problems. Our algorithms construct small enough coresets to approximate the cost of the relational clustering.

A weighted set $\coreset\subseteq \Re^d$ is an $\eps$-coreset for the relational $k$-median clustering problem on $\Q(\I)$ if for any set of $k$ centers $Y\subset \Re^d$,
\begin{equation}
\label{def:coreset-kmedian}
(1-\eps)\kmedian_Y(\Q(\I))\leq \kmedian_{Y}(\coreset)\leq (1+\eps)\kmedian_Y(\Q(\I)).    
\end{equation}

Similarly, a weighted set $\coreset$ is an $\eps$-coreset for the relational $k$-means clustering problem, if for any set of $k$ centers $Y\subset \Re^d$,
\begin{equation}
\label{def:coreset-kmeans}
(1-\eps)\kmeans_Y(\Q(\I))\leq \kmeans_{Y}(\coreset)\leq (1+\eps)\kmeans_Y(\Q(\I)).    
\end{equation}

\subsection{Data structures for aggregation queries}
In this subsection, we show how we can combine known results in database theory to answer aggregation queries in linear time with respect to the size of the database.

Let $\mathcal{R}$ be an axis-parallel hyper-rectangle in $\Re^d$.
The goal is i) count the number of tuples $|\Q(\I)\cap \mathcal{R}|$, and ii) sample uniformly at random from $\Q(\I)\cap \mathcal{R}$.
The axis-parallel hyper-rectangle $\mathcal{R}$ is defined as the product of $d$ intervals over the attributes, i.e., $\mathcal{R}=\{I_1\times \ldots\times I_d)$, where $I_i=[a_i,b_i]$ for $a_i,b_i\in \Re$.
Hence, $\mathcal{R}$ defines a set of $d$ linear inequalities over the attributes, i.e., a tuple $t$ lies in $\mathcal{R}$ if and only if $a_j\leq \pi_{A_j}(t)\leq b_j$ for every $A_j\in\allattr$.
Let $p$ be a tuple in a relation $R_i$.
If $a_j\leq \pi_{A_j}(p)\leq b_j$ for every $A_j\in \allattr_i$, then we keep $p$ in $R_i$. Otherwise, we remove it. The set of surviving tuples is exactly the set of tuples that might lead to join results in $\mathcal{R}$. Let $\I'\subseteq \I$ be the new database instance such that $\Q(\I')=\Q(\I)\cap\mathcal{R}$. The set $\I'$ is found in $O(N)$ time.
Using Yannakakis algorithm~\cite{yannakakis1981algorithms} we can count $|\Q(\I')|$ in $O(N\log N)$ time and using~\cite{zhao2018random} we can sample $z$ tuples from $\Q(\I')$ in $O((N+z)\log N)$ time. Using hashing, we can improve the running time of the algorithms to $O(N)$ and $O(N+z\log N)$, respectively.
\vspace{-0.5em}
\begin{lemma}
\label{lem:Rects}
    Let $\mathcal{R}$ be a rectangle in $\Re^d$. There exists an algorithm $\mathsf{CountRect}(\Q, \I,\mathcal{R})$ to count $|\Q(\I)\cap \mathcal{R}|$ in $O(N\log N)$ time or $O(N)$ time with high probability. Furtheremore, there exists an algorithm $\mathsf{SampleRect}(\Q, \I,\mathcal{R}, z)$ to sample $z$ samples from $\Q(\I)\cap \mathcal{R}$ in $O((N+z)\log N)$  time or $O(N+z\log N)$ time with high probability.
\end{lemma}
Let $B\subseteq \allattr$ be a subset of the attributes. The result in Lemma~\ref{lem:Rects} can also be used (straightforwardly) to compute $|\widebar{\pi}_{B}(\Q(\I))\cap \mathcal{R}|$ or sample $z$ samples from $\widebar{\pi}_{B}(\Q(\I))\cap \mathcal{R}$ in the same running time, since $|\widebar{\pi}_{B}(\Q(\I))\cap \mathcal{R}|=|\Q(\I)\cap \mathcal{R}|$.

\subsection{High level ideas}
Before we start with the technical sections of the paper, we give high level ideas of our methods.

In Section~\ref{sec:coresets}, we assume that a set of centers $\kapprox$ for the relational $k$-median clustering is given such that $\kmedian_{\kapprox}(\Q(\I))$ is within a constant factor from the optimum $k$-median error with $|\kapprox|>k$.
Using $\kapprox$, we construct a coreset $\coreset$ for the relational $k$-median problem.
Then we run $\kmedianAlg_\gamma$ (or $\DkmeansAlg_\gamma$ for the discrete version) on $\coreset$ to obtain the final solution $\ret$. While all previous papers on relational clustering~\cite{chen2022coresets, curtin2020rk, moseley2021relational} construct weighted coresets, the coresets that were used are sub-optimal. Instead, we use and modify the coreset construction from~\cite{har2004coresets}, which is a more optimized coreset for the $k$-median (and $k$-means) clustering under the Euclidean metric, in the standard computational setting. Of course, using such a coreset comes at a cost.
First, it is not clear how to construct a set $\kapprox$ with the desired properties on relational data, efficiently. In Section~\ref{sec:kmedianAlg} we show how to construct $\kapprox$ designing a hierarchical method over the attributes $\allattr$.
Second, it is not straightforward to construct the coreset in~\cite{har2004coresets} on relational data, given $\kapprox$. In fact, the coreset construction in~\cite{har2004coresets} cannot be applied (efficiently) in the relational setting. Hence, inspired by ~\cite{har2004coresets}, we design a novel small coreset based on $\kapprox$ and show that it can be applied to relational data.

More specifically, for any center $x_i\in\kapprox$, in \cite{har2004coresets}, the authors first compute all points in the dataset that have $x_i$ as the closest center in $\kapprox$. Let $P_i$ be this set of points. Then, they construct a grid around $x_i$, and from every cell $\square$ in the grid, they add one representative point from $P_i\cap \square$ with weight $|P_i\cap \square|$.
Unfortunately, in our setting we cannot compute $P_i$ and/or $|P_i\cap \square|$ efficiently.

We resolve the two issues as follows.
First, we construct a grid around $x_i$, but instead of computing all tuples that have $x_i$ as the closest center and process all cells, we check whether a cell $\square$ is close enough to the center $x_i$ (Equation~\eqref{eq:cond} in the next section). If not then we skip the cell. If yes, then we take a representative tuple from $\square$ and we set its weight to be the number of tuples in $\square$ that do not lie in a different cell that has been already processed by our algorithm. Of course, the weight of a representative tuple now is not the same as in~\cite{har2004coresets} and we might count tuples even outside of $P_i$, however with a careful analysis we make sure that the overall error is still bounded.
The second problem though, still remains; we need to count the number of tuples in $\square$ excluding the cells we have already visited from previous centers in $\kapprox$ (notice that the grid cells of two different centers in $\kapprox$ might intersect). We propose two methods to achieve it. In Section~\ref{subsec:detmed} we give a deterministic method that constructs the arrangement of the complement of the visited cells and we use Lemma~\ref{lem:Rects} to count the number of tuples, exactly. However, this algorithm requires $\Omega(|\kapprox|^dN)$ time. In Section~\ref{subsec:medrand} we give a more involved and faster randomized approximation algorithm to count the number of tuples based on sampling.



\section{From many centers to exactly $k$ centers}
\label{sec:coresets}
We describe an algorithm that constructs a coreset for the relational $k$-median clustering over the (multi-set) projection of $\Q(\I)$ on an arbitrary subset of attributes.
Let $\allattr_u\subseteq \allattr$ be a subset of the attributes. Let $\Q_u(\I)$ be the \emph{multi-set} $\Q_u(\I):=\widebar{\pi}_{\allattr_u}(\Q(\I))$. 
Notice that $|\Q_u(\I)|=|\Q(\I)|$ and its size can be computed in $O(N)$ time using Yannakakis algorithm. Let $n=|\Q_u(\I)|$ and $d_u=|\allattr_u|$.
Assume that $\kapprox$ is a set of points in $\Re^{d_u}$ such that $\kmedian_{\kapprox}(\Q_u(\I))\leq \alpha \cdot\kmedian_{\opt(\Q_u(\I))}(\Q_u(\I))$, where $\alpha>1$ is a constant.
Notice that $$\kmedian_{\kapprox}(\Q_u(\I))=\sum_{t\in \Q(\I)}\dist(\pi_{\allattr_u}(t), \kapprox).$$
We also assume that $r$ is a real number such that $\kmedian_{\kapprox}(\Q_u(\I))\leq r\leq \alpha \cdot\kmedian_{\opt(\Q_u(\I))}(\Q_u(\I))$.
In fact, it follows from the following sections that we can also assume $\frac{r}{(1+\eps)\sqrt{2}}\leq \kmedian_{\kapprox}(\Q_u(\I))$.
Note that we do not assume anything about the size of $\kapprox$. In the next section, we show that we can always consider $|\kapprox|=O(k^2)$. 

In this section we propose two algorithms that take as input $\kapprox$ and $r$ and return a set $\ret\subset \Re^{d_u}$ of cardinality $|\ret|=k$ such that $\kmedian_\ret(\Q_u(\I))\leq (1+\eps)\gamma \kmedian_{\opt(\Q_u(\I))}(\Q_u(\I))$, i.e., $\ret$ is a $(1+\eps)\gamma$-approximation for the $k$-median problem in $\Q_u(\I)$. They also return a number $r_u$ such that $\kmedian_\ret(\Q_u(\I))\leq r_u\leq (1+\eps)\gamma \kmedian_{\opt(\Q_u(\I))}(\Q_u(\I))$. We note that given a set of centers $S$, we cannot compute the error $\kmedian_S(\Q_u(\I))$ efficiently, hence $r_u$ is needed to estimate the error of the clustering.
We also note that if $\allattr_u=\allattr$, then the returned set $\ret$ is an $(1+\eps)\gamma$ approximation solution for the relational $k$-median clustering problem on $\Q(\I)$. 
For the discrete relational $k$-median clustering, we return $\ret\subseteq \Q(\I)$ and $r_u$ such that, $\kmedian_\ret(\Q_u(\I))\leq r_u\leq (2+\eps)\gamma \kmedian_{\opt(\Q_u(\I))}(\Q_u(\I))$.

The first algorithm is a slow deterministic algorithm that runs in $\Omega(|\kapprox|^{d_u+1}N)$ time. The second algorithm is a faster randomized algorithm that runs in time roughly $O(|\kapprox|N)$ time.
We mostly focus on the geometric version of the relational $k$-median clustering, however, we always highlight the differences with the discrete version.

\subsection{Slow deterministic algorithm}
\label{subsec:detmed}

\begin{algorithm}[t]
\caption{{\sc RelClusteringSlow}$(\Q, \I, \allattr_u, \kapprox, \alpha, r, \eps)$}
\label{alg:Determ}
$\eps'=\eps/4$\;
$n=|\Q_u(\I)|=|\Q(\I))|\quad\quad$ (using Lemma~\ref{lem:Rects})\;
$\Phi=\frac{r}{\alpha n}$;
$\quad G=\emptyset$; $\quad \coreset=\emptyset$\;
\ForEach{$x_i\in \kapprox$}{
    \ForEach{$j=0,1,\ldots, 2\log(\alpha n)$}{
        $Q_{i,j}\leftarrow$ axis parallel square with length $2^j\Phi$ centered at $x_i$\;
        $V_{i,j}=Q_{i,j}\setminus Q_{i,j-1}$;
        $\quad \widebar{V}_{i,j}\leftarrow$ grid of side length $\eps'2^j\Phi/(10\alpha d_u)$ in $V_{i,j}$\;
    }
    $\widebar{V}_i=\bigcup_{j}\widebar{V}_{i,j}$\;
    \ForEach{$\square\in \widebar{V}_i$}{
        \If{$\dist(x_i,\square)\leq \dist(\kapprox,\square)+\diameter(\square)$}{
            $\mathsf{Arr}(G)\leftarrow$ arrangement of $G$;
$\quad \widebar{\mathsf{Arr}}(G)\leftarrow$ complement of $\mathsf{Arr}(G)$\;
$\mathsf{Arr}'(G)\leftarrow$ partition of $\widebar{\mathsf{Arr}}(G)$ into hyper-rectangles\;
$K_\square=0$\;
\ForEach{$\mathcal{R}\in \mathsf{Arr}'(G)$}{
    $\square_{\mathcal{R}}=\square\cap \mathcal{R}$;
    $\quad K_\square=K_\square+\textsf{CountRect}(\Q,\I,\square_{\mathcal{R}})\quad\quad$ (using Lemma~\ref{lem:Rects})\;
}
            \If{$K_\square>0$}{
                $s_\square\leftarrow$ arbitrary tuple in $\Q_u(\I)\cap (\square\setminus G)$\;
                $w(s_\square)=K_\square$\;
                $\coreset\leftarrow\coreset\cup\{s_\square\}$\;
            }
            $G\leftarrow G\cup\{\square\}$\;
        }
    }
    }
    $\ret=\kmedianAlg_\gamma(\coreset)\quad\quad$ (or $\ret=\DkmedianAlg_\gamma(\coreset)$ for the discrete version)\;
    $r_u=\frac{1}{1-\eps'}\kmedian_\ret(\coreset)$\;
\Return $(\ret, r_u)$\;
\end{algorithm}

For the slow deterministic algorithm we construct a hierarchical grid around every center $x_i\in \kapprox$. Then, for every cell that is close enough to $x_i$, we compute the number of join results in $\Q_u(\I)$ inside the cell that do not lie inside another cell previously processed by our algorithm. We count the number of join results using Lemma~\ref{lem:Rects}.

\paragraph{Algorithm}
The pseudocode of this algorithm is shown in Algorithm~\ref{alg:Determ}.
We first set $\eps'=\eps/4$ and define $\Phi=\frac{r}{\alpha n}$ as a lower bound estimate of the average radius $\frac{\kmedian_{\opt(\Q_u(\I))}(\Q_u(\I))}{n}$.
For every point $x_i\in \kapprox$ we construct an exponential grid around $x_i$. Let $Q_{i,j}$ be an axis parallel square with side length $\Phi\cdot 2^j$ centered at $x_i$, for $j=0,1,\ldots, 2\log(\alpha n)$. Let $V_{i,0}=Q_{i,0}$ and let
$V_{i,j}=Q_{i,j}\setminus Q_{i,j-1}$. We partition $V_{i,j}$ into a grid $\widebar{V}_{i,j}$ of side length $\eps'\Phi 2^j/(10\alpha d_u)$. Let $\widebar{V}_i=\bigcup_j \widebar{V}_{i,j}$.
\begin{figure}[ht]
    \centering
    \begin{minipage}{0.49\textwidth}
        \centering
        \begin{tikzpicture}[scale=1.5]
    \draw[step=0.1cm,black, thick] (0,0) grid (0.4,0.4);

    \draw[step=0.2cm,black,thick] (-0.4,-0.4) grid (0.8,0.8);

     \draw[step=0.4cm,black,thick] (-1.2,-1.2) grid (1.6,1.6);

    \draw[thick] (-1.2,-1.2) rectangle (1.6,1.6);
    
    \fill[red] (0.2,0.2) circle (1pt);
\end{tikzpicture}
\caption{The grid construction $\bar{V}_i$ around the red point $x_i\in \kapprox$.}
\label{fig:grid}
    \end{minipage}\hfill
    \begin{minipage}{0.49\textwidth}
        \centering
        \begin{tikzpicture}[scale=0.6]
    \draw[thick] (5,5) rectangle (7,7);
    \node[right] at (7,6.8) {$\square_1$};
    \draw[thick] (11.5,1.5) rectangle (13,3);
    \node[right] at (13,2.8) {$\square_2$};
    
    \fill[black] (10,2) circle (4pt);
    \node[above] at (10.2,2.1) {$x_i$};
    \fill[black] (5.5,4) circle (4pt);
    \node[above] at (6.1,4) {$x_{i+1}$};
    \fill[black] (12.5,4) circle (4pt);
    \node[above] at (12.8,4) {$x_{i+2}$};
    \fill[black] (9.5,6) circle (4pt);
    \node[above] at (9.8,6) {$x_{i+3}$};

    \draw[red, dashed, thick] (10,2) -- (7,5);
    \draw[red, dashed, thick] (5.5,4) -- (5.5,5);
    \draw[red, dashed, thick] (5,5) -- (7,7);

    \draw[blue, dashed, thick] (10,2) -- (11.5,2);
     \draw[blue, dashed, thick] (12.5,4) -- (12.5,3);
     \draw[blue, dashed, thick] (11.5,1.5) -- (13,3);
\end{tikzpicture}
\caption{Let $\square_1, \square_2\in \bar{V}_i$. It holds that $\dist(x_i,\square_1)>\dist(x_{i+1},\square_1)+\diameter(\square_1)$ so $\square_1$ is not processed by the algorithm. $x_{i+2}$ is the closest center to $\square_2$, i.e., $\dist(\kapprox,\square_2)=\dist(x_{i+2},\square_2)$ and it holds $\dist(x_i,\square_2)<\dist(x_{i+2},\square_2)+\diameter(\square_2)$, so $\square_2$ is processed by the algorithm. The red (blue) dashed segments represent the distances of $x_i$ to $\square_1$ ($\square_2$), $x_{i+1}$ ($x_{i+2}$) to $\square_1$ ($\square_2$), and the diameter of $\square_1$ ($\square_2$).}
\label{fig:condition}
    \end{minipage}
\end{figure}

An example of the grid construction can be seen in Figure~\ref{fig:grid}.
The construction so far is similar to the exponential grid in~\cite{har2004coresets}. However, from now on the algorithm and the analysis is different.
Let $G=\emptyset$ be an empty set of grid cells. Let also $\coreset=\emptyset$ be an empty point set in $\Re^{d_u}$ and $w$ be a weight function that we are going to define on $\coreset$.

For every $x_i\in X$, and for every cell $\square\in \widebar{V}_i$ we repeat the following steps. Let $\diameter(\square)$ be the diameter of $\square$. If
\vspace{-1em}
\begin{equation}
    \label{eq:cond}
    \dist(x_i,\square)\leq \dist(X,\square)+\diameter(\square),
\end{equation}
(an example of applying Equation~\eqref{eq:cond} is shown in Figure~\ref{fig:condition}) then we proceed as follows. Let $G_\square$ be the set $G$ just before the algorithm processes the cell $\square$.
The goal is to identify $|\Q_u(\I)\cap (\square \setminus G_\square)|$ and if $|\Q_u(\I)\cap (\square \setminus G_\square)|>0$ then add a representative point $s_{\square}\in \Q_u(\I)\cap (\square \setminus G_\square)$ in $\coreset$ with weight $w(s_\square)=|\Q_u(\I)\cap (\square \setminus G_\square)|$.
Next, we construct the arrangement of $G_\square$.
The arrangement~\cite{agarwal2000arrangements, halperin2017arrangements} of $G_\square$, denoted $\textsf{Arr}(G_\square)$, is a partitioning of $G_\square$ into
rectangular contiguous regions, such that for every region $\mathsf{reg}$ in the arrangement, $\mathsf{reg}$ lies in the same subset of $G_\square$.
Let $\widebar{\mathsf{Arr}}(G_\square)$ be the complement of $\textsf{Arr}(G_\square)$, and let $\textsf{Arr}'(G_\square)$ be a partition of $\widebar{\mathsf{Arr}}(G_\square)$ into hyper-rectangles.
For each rectangle $\mathcal{R} \in \textsf{Arr}'(G_\square)$, we compute $\square_{\mathcal{R}}=\square\cap \mathcal{R}$, which is also a hyper-rectangle in $\Re^{d_u}$. Using the $\textsf{CountRect}(\cdot)$ procedure form Lemma~\ref{lem:Rects} we compute $K_{\square}=\sum_{\mathcal{R}\in\textsf{Arr}'(G)}|\square_{\mathcal{R}}\cap \Q_u(\I)|$. If $K_{\square}>0$ we also get an arbitrary point $s_{\square}$ from $\bigcup_{\mathcal{R}\in \textsf{Arr}'(G)} \square_{\mathcal{R}}\cap \Q_u(\I)$. We add $s_\square$ in $\coreset$ and we set its weight $w(s_\square)=K_{\square}$. We say that all tuples in $\Q_u(\I)\cap (\square \setminus G_\square)$ are \emph{assigned} to $s_\square$, and we add $\square$ in $G$.
An example of the arrangement can be seen in Figure~\ref{fig:arrangement}.
\begin{figure}[ht]
\begin{minipage}{0.49\textwidth}
    \centering
    \begin{tikzpicture}[scale=0.25]
    \draw[thick] (0,0) rectangle (10,10);
    \draw[very thick,blue] (-0.5,-2.5) rectangle (6,4);
    \draw[very thick,blue] (4.5,-0.5) rectangle (7,2);
     \draw[very thick,blue] (5,7) rectangle (10.5,12.5);
     
      \draw[red, dashed,  thick] (4.9,4.1) -- (4.9,9.9);
      \draw[red, dashed, thick] (0.1,4.1) -- (0.1,9.9);
      \draw[red, dashed, thick] (0.1,4.1) -- (9.9,4.1);
      \draw[red, dashed, thick] (0.1,9.9) -- (4.9,9.9);
      \draw[red, dashed, thick] (0.1,6.9) -- (9.9,6.9);
      \draw[red, dashed, thick] (7.1,0.1) -- (7.1,6.9);
      \draw[red, dashed, thick] (7.1,0.1) -- (9.9,0.1);
       \draw[red, dashed, thick] (6.1,2.1) -- (9.9,2.1);
       \draw[red, dashed, thick] (9.9,6.9) -- (9.9,0.1);
       \draw[red, dashed, thick] (6.1,2.1) -- (6.1,6.9);

    \node[right] at (3.2,11) {{\color{blue}$G$}};
\end{tikzpicture}
\caption{The black square is a cell $\square\in\bar{V}_i$. The blue squares define the set of cells in $G_\square$ (cells already processed by the algorithm) that intersect $\square$. The red dashed segments show the arrangement of the complement of $G_{\square}$ in $\square$, i.e., the red dashed segments show $\square\cap \textsf{Arr}'(G_\square)$.}
\label{fig:arrangement}
    \end{minipage}\hfill
    \begin{minipage}{0.49\textwidth}
        \centering
        \begin{tikzpicture}[scale=0.25]
    \draw[thick] (0,0) rectangle (10,10);
    \draw[very thick,blue] (-0.5,-2.5) rectangle (6,4);
    \draw[very thick,blue] (4.5,-0.5) rectangle (7,2);
     \draw[very thick,blue] (5,7) rectangle (10.5,12.5);

    \fill[blue] (1,3) circle (4pt);
    \fill[blue] (2.5,3.7) circle (4pt);
    \fill[blue] (4.8,1.2) circle (4pt);
    \fill[blue] (6.2,8.5) circle (4pt);

     \fill[black] (0.5,7.5) circle (4pt);
     \fill[black] (1.3,9) circle (4pt);
     \fill[black] (3.1,5.2) circle (4pt);
     \fill[black] (3.8,8.2) circle (4pt);
     \fill[black] (6.5,5.8) circle (4pt);
     \fill[black] (8,2.2) circle (4pt);
     \fill[black] (9,6.5) circle (4pt);
     \fill[black] (7,6) circle (4pt);
     \node[right] at (3.2,11) {{\color{blue}$B$}};
\end{tikzpicture}
\caption{The black square is a cell $\square\in\bar{V}_i$. The blue squares define the set of heavy cells in $B_\square$ that intersect $\square$. The points represent the set of samples $H_{\square}$. Blue points are the samples in $B_\square$ while black points are the samples in $\square\setminus B_\square$. Hence, $g_\square=8$ and $M=|H_\square|=12$. If $\tau=0.3$ then $g_\square/M\geq 2\cdot \tau$, and a black point is selected as $s_\square$ with weight $\frac{8}{12}\cdot\frac{n_\square}{1-\eps'}$.}
\label{fig:sample}
    \end{minipage}
\end{figure}

On the other hand, if the condition~\eqref{eq:cond} is not satisfied, then we skip $\square$ and we continue with the next cell in the exponential grid.


In the end, after repeating the algorithm for each $x_i\in \kapprox$, we get a weighted set $\coreset$. We run the standard algorithm for the weighted $k$-median problem $\kmedianAlg_\gamma$ (or $\DkmedianAlg_\gamma$ for the discrete $k$-median problem) on $\coreset$ to get a set of $k$ centers $\ret$, and we return $\ret$ as the answer. Furthermore, we return $r_u=\frac{1}{1-\eps'}\kmedian_{\ret}(\coreset)$.

\paragraph{Correctness}
We first show that every tuple in $\Q_u(\I)$ is assigned to a point in $\coreset$. Then, we show that $\coreset$ is an $\eps'$-coreset for $\Q_u(\I)$. Next, we show that $\ret$ is a good approximation of the optimum $k$-median solution and $r_u$ is a good approximation of $\kmedian_\ret(\Q_u(\I))$. Finally, we show that $\kmedian_\ret(\Q_u(\I))\leq r_u\leq (1+\eps)\gamma \kmedian_{\opt(\Q_u(\I))}(\Q_u(\I))$.
All missing proofs can be found in Appendix~\ref{appndx:slowAlg}.

\begin{lemma}
\label{lem:tech1}
    Every tuple $t\in \Q_u(\I)$ is assigned to a point in $\coreset$. Furthermore, the number of tuples in $\Q_u(\I)$ that are assigned to a point $s\in \coreset$ is $w(s)$.
\end{lemma}
\begin{proof}
    Let $x_i\in \kapprox$ be the center that is closest to $t$. By definition, there will be a cell $\square$ defined by the exponential grid around $x_i$ that contains $t$, since $\phi(t, x_i) \leq \alpha n \Phi$. We show that for the cell $\square$ the condition~\eqref{eq:cond} holds. Let $x_j$ be the center such that $\dist(x_j,\square)=\dist(\kapprox, \square)$. We have $\dist(x_i,\square)\leq \dist(x_i,t)\leq \dist(x_j,t)\leq \dist(\kapprox,\square)+\diameter(\square)$. Hence, $t$ is assigned in $s_{\square}$.
    The second part of the lemma holds by definition.
\end{proof}
For a tuple $t\in\Q_u(\I)$, let $\square_t$ be the cell from the exponential grid defined by the algorithm such that 
$\square_t$ is the first cell processed by the algorithm that contains $t$, i.e.,
$t\in\Q_u(\I)\cap (\square_t\setminus G_{\square_t})$.
Next, we denote the assignment of each tuple $t\in \Q_u(\I)$ to a point in $\coreset$ by $\sigma$, i.e., $\sigma(t):=s_{\square_t}\in \coreset$.
Let $i(t)$ be the index such that $\square_t\in \widebar{V}_{i(t)}$.
By definition, notice that $x_{i(t)}\in X$ is the center visited by our algorithm at the moment that $t$ is assigned to $\sigma(t)$.

\begin{lemma}
    \label{lem:coreset1}
    $\coreset$ is an $\eps'$-coreset for $\Q_u(\I)$.
\end{lemma}
\begin{proof}
    Let $Y$ be an arbitrary set of $k$ points in $\Re^{d_u}$. The error is defined as $\mathcal{E}=|\kmedian_Y(\Q_u(\I))-\kmedian_Y(\coreset)|\leq \sum_{t\in \Q_u(\I)}|\dist(t,Y)-\dist(\sigma(t),Y)|$. By the triangle inequality, $\dist(t,Y)\leq \dist(\sigma(t),Y)+\dist(t,\sigma(t))$ and $\dist(\sigma(t),Y)\leq \dist(t,Y)+\dist(t,\sigma(t))$. Hence, $|\dist(t,Y)-\dist(\sigma(t),Y)|\leq \dist(t,\sigma(t))$.
    We have, 
    \begin{align*}
    \mathcal{E}\leq \sum_{t\in\Q_u(\I)}\dist(t,\sigma(t))= \sum_{t\in\Q_u(\I), \dist(t,x_{i(t)})\leq \Phi}\dist(t,\sigma(t))\quad\quad+\sum_{t\in\Q_u(\I), \dist(t,x_{i(t)})> \Phi}\dist(t,\sigma(t)).
    \end{align*}
    For $\dist(t,x_{i(t)})\leq \Phi$, by the construction of the exponential grid, we have $\dist(t,\sigma(t))\leq \diameter(\square_t)\leq\frac{\eps'}{10\alpha}\Phi$, hence
    \vspace{-0.8em}$$\sum_{t\in\Q_u(\I), \dist(t,x_{i(t)})\leq \Phi}\dist(t,\sigma(t))\leq \frac{\eps'}{10\alpha}n\Phi
    \leq \frac{\eps'}{10\alpha}\kmedian_{\opt(\Q_u(\I))}(\Q_u(\I)).$$

    For $\dist(t,x_{i(t)})>\Phi$, by the construction of the exponential grid we have $\dist(t,\sigma(t))\leq \diameter(\square_{t}) \leq \frac{\eps'}{10\alpha}\dist(t,x_{i(t)})$. 
    Notice that $\dist(x_{i(t)},\square_t)\leq \dist(X,\square_t)+\diameter(\square_t)$, by condition~\eqref{eq:cond}.
    We have,
    \begin{align*}
        \dist(t,x_{i(t)})&\leq \dist(x_{i(t)},\square_{t})+\diameter(\square_{t})\leq \dist(\kapprox,\square_t)+2\diameter(\square_t)\leq \dist(t,\kapprox)+\frac{2\eps'}{10\alpha}\dist(t,x_{i(t)})\\&\Leftrightarrow \dist(t,x_{i(t)})\leq \frac{1}{1-\frac{2\eps'}{10\alpha}}\dist(t,X)\leq (1+\frac{4\eps'}{10\alpha})\dist(t,\kapprox),
    \end{align*}
    so $\dist(t,\sigma(t))\leq \frac{\eps'}{10\alpha}(1+\frac{4\eps'}{10\alpha})\dist(t,\kapprox)$.
    Then, we get
    \begin{align*}
    \mathcal{E}&\leq \frac{\eps'}{10\alpha}\kmedian_{\opt(\Q_u(\I))}(\Q_u(\I))+\frac{\eps'}{10\alpha}(1+\frac{4\eps'}{10\alpha})\kmedian_{\kapprox}(\Q_u(\I))\leq (\frac{2\eps'}{10}+\frac{4(\eps')^2}{100\alpha})\kmedian_{\opt(\Q_u(\I))}(\Q_u(\I))\\&\leq \eps' \kmedian_{\opt(\Q_u(\I))}(\Q_u(\I)).
    \end{align*}
    This implies that $|\kmedian_Y(\Q_u(\I))-\kmedian_Y(\coreset)|\leq \eps'\kmedian_Y(\Q_u(\I))$. 
\end{proof}

From the previous lemma, we conclude with the main result establishing the correctness of the algorithm.
\begin{lemma}
\label{lem:correct1}
    If $\kmedianAlg_\gamma$ is used, then $\ret\subset \Re^d$ and $\kmedian_\ret(\Q_u(\I))\leq r_u\leq (1+\eps)\gamma \kmedian_{\opt(\Q_u(\I))}(\Q_u(\I))$. If $\DkmedianAlg_\gamma$ is used, then $\ret\subseteq \Q_u(\I)$ and  $\kmedian_\ret(\Q_u(\I))\leq r_u\leq (2+\eps)\gamma \kmedian_{\optD(\Q_u(\I))}(\Q_u(\I))$.
\end{lemma}
\begin{proof}
    By the definition of $\kmedianAlg_\gamma$ and Lemma~\ref{lem:coreset1}, we have $\kmedian_\ret(\coreset)\leq \gamma\kmedian_{\opt(\coreset)}(\coreset)\leq \gamma\kmedian_{\opt(\Q_u(\I))}(\coreset)\leq (1+\eps')\gamma\kmedian_{\opt(\Q_u(\I))}(\Q_u(\I))$,
    and $\kmedian_\ret(\coreset)\geq (1-\eps')\kmedian_\ret(\Q(\I))$, so
    \begin{align*}
    \kmedian_\ret(\Q_u(\I))&\leq \frac{1}{1-\eps'}\kmedian_{\ret}(\coreset)=r_u\leq 
    \gamma\frac{1}{1-\eps'}\kmedian_{\opt(\coreset)}(\coreset) 
    \leq \frac{1+\eps'}{1-\eps'}\gamma \kmedian_{\opt(\Q_u(\I))}(\Q_u(\I))\\&\leq (1+4\eps')\gamma \kmedian_{\opt(\Q_u(\I))}(\Q_u(\I))=(1+\eps)\gamma \kmedian_{\opt(\Q_u(\I))}(\Q_u(\I)).
    \end{align*}
    The proof using the $\DkmedianAlg_\gamma$ algorithm is shown in Appendix~\ref{appndx:slowAlg}.
\end{proof}

\paragraph{Running time}
\begin{lemma}
\label{lem:size}
$|\coreset|=O(|X|\eps^{-d_u}\log N)$.
\end{lemma}
\begin{proof}
    For each $x_i\in X$ there are $2\log(\alpha N)=O(\log N)$ boxes  $Q_{i,j}$. For each $V_{i,j}$, we construct $O(\eps^{-d_u})$ cells. In the worst case, $\coreset$ contains a point for every cell.
\end{proof}
In total, the algorithm visits $O(|\kapprox|\eps^{-d_u}\log N)$ cells. For each cell $\square$, we check the condition~\eqref{eq:cond} in $O(|\kapprox|)$ time and  we compute the arrangement $\textsf{Arr}(G_\square)$ in $O(|\kapprox|^{d_u}\eps^{-d_u^2}\log^{d_u}N)$ time. For each hyper-rectangle in the complement, we run a query as in Lemma~\ref{lem:Rects} in $O(N\log N)$ time. Overall, $\coreset$ is constructed in $O(|\kapprox|^{d_u+1}\eps^{-d_u^2-d_u}N\log^{d_u+2} N)$ time. Finally, we get $\ret$ in $O(\timeMedian_\gamma(|\kapprox|
\eps^{-d_u}\log N))$ time.

\begin{theorem}
    \label{coreset1}
    Let $\I$ be a database instance with $N$ tuples, $\Q$ be an acyclic join query over a set of attributes $\allattr$ and $\allattr_u\subseteq \allattr$. Given a set $X\subset \Re^d$, a constant $\alpha$ such that $\kmedian_X(\Q_u(\I))\leq \alpha \kmedian_{\opt(\Q_u(\I))}(\Q_u(\I))$, and a constant parameter $\eps\in(0,1)$, there exists an algorithm that computes a set $\ret\subset \Re^d$ of $k$ points and a number $r_u$ in $O(|\kapprox|^{d_u+1}N\log^{d_u+2} N+\timeMedian_\gamma(|\kapprox|
\log N))$ time such that $\kmedian_{\ret}(\Q_u(\I))\leq r_u\!\leq\! (1+\eps)\gamma\kmedian_{\opt(\Q_u(\I))}(\Q_u(\I))$. There also exists an algorithm that computes a set $\ret\subseteq \Q_u(\I)$ of $k$ points and a number $r_u$ with the same running time, such that $\kmedian_{\ret}(\Q_u(\I))\leq r_u\!\leq\! (2+\eps)\gamma\kmedian_{\optD(\Q_u(\I))}(\Q_u(\I))$.
\end{theorem}
\subsection{Fast randomized algorithm}
\label{subsec:medrand}
As we show in Section~\ref{subsec:algmedian}, the algorithm in the previous section can be used to get a constant approximation for the relational $k$-median problem. However, the running time is $\Omega(|\kapprox|^{d_u+1}N)$. As we will see in the next section $|\kapprox|=k^2$, leading to an $\Omega(k^{2d_u+2}N)$ algorithm. 
In this section, we propose a more involved randomized algorithm that improves the factor $|\kapprox|^{d_u+1}N$ to only $|\kapprox|\cdot N$.
Undoubtedly, the expensive part of the deterministic algorithm is the cardinality estimation $|\Q_u(\I)\cap (\square\setminus G_\square)|$. Next, we design a faster algorithm to overcome this obstacle.
\begin{algorithm}[t]
\caption{{\sc RelClusteringFast}$(\Q, \I, \allattr_u, \kapprox, \alpha, r, \eps)$}
\label{alg:Random}
    $\eps'=\eps/34$\;
    Lines 2--3 from Algorithm~\ref{alg:Determ}\;
    \ForEach{$x_i\in X$}{
    Lines 5--8 from Algorithm~\ref{alg:Determ}\;
    $B=\emptyset$\;
    $\tau=\frac{1}{16|\kapprox|(\eps')^{-d_u-1}\log N}$\;
    $M=\frac{3}{(\eps')^2\tau}\log(2N^{10d})$\;
    \ForEach{$\square\in \widebar{V}_i$}{
        \If{$\dist(x_i,\square)\leq \dist(\kapprox,\square)+\diameter(\square)$}{
        $H_\square=\widebar{\pi}_{\allattr_u}\left(\textsf{SampleRect}(\Q,\I,\square,M)\right)\quad\quad$ (Using Lemma~\ref{lem:Rects})\;
        $g_\square=|H_\square\setminus(B\cap H_\square)|$\;
            \If{$\frac{g_\square}{M}\geq 2\tau$}{
                $s_\square\leftarrow$ arbitrary tuple in $H_\square\setminus B$\;
                $n_\square=\textsf{CountRect}(\Q,\I,\square)\quad\quad$ (Using Lemma~\ref{lem:Rects})\;
                $w(s_\square)=\frac{1}{1-\eps'}\cdot\frac{g_\square}{M}\cdot n_\square$\;
                $\coreset\leftarrow\coreset\cup\{s_\square\}$;
                $\quad B\leftarrow B\cup\{\square\}$\;
            }
        }
    }
    }
    $\ret=\kmedianAlg_\gamma(\coreset)\quad\quad$ (or $\ret=\DkmedianAlg_\gamma(\coreset)$ for the discrete version)\;
    $r_u=\frac{1+4\eps'}{1-9\eps'}\kmedian_\ret(\coreset)$\;
    \Return $(\ret, r_u)$\;
\end{algorithm}
The algorithm constructs exactly the same exponential grid as described above. However, in this algorithm, we use a more involved approach to estimate the weights $w(s_\square)$ faster using random uniform sampling.
We use the same notation as in the previous subsection. Let $\coreset=\emptyset$ and let $\Phi$ as defined above.
In this algorithm, we will characterize each cell we visit as \emph{heavy} or \emph{light}.
Let $B$ denote the set of the processed heavy cells. So, we initialize with $B=\emptyset$.


\mparagraph{Algorithm}The pseudocode of the algorithm is shown in Algorithm~\ref{alg:Random}.
We set $\eps'=\eps/34$. For each $x_i\in X$ and for each cell $\square\in \widebar{V}_i$,
we check condition~\eqref{eq:cond}. If it is satisfied, then we process $\square$. Otherwise, we skip it and continue with the next cell.
Let $\square$ be a cell in the exponential grid that the algorithm processes. We compute $n_{\square}=|\Q_u(\I)\cap \square|$ and we set $\tau=\frac{1}{16|\kapprox|(\eps')^{-d_u-1}\log N}$.
Using the algorithm from Lemma~\ref{lem:Rects}, we sample with replacement a multi-set $H_{\square}$ of $M=\frac{3}{(\eps')^2\tau} \log(2N^{10d})=O(|\kapprox|(\eps')^{-d_u-3}\log^2 N)$ points from $\Q_u(\I)\cap \square$ and we set $g_\square=|H_\square\setminus (B\cap H_{\square})|$, i.e., the number of samples that are not currently contained in heavy cells we processed. If $\frac{g_\square}{M}\geq 2\tau$, then let $s_{\square}$ be any of the sampled points in $H_\square\setminus B$ as the representative point of $\square$. We add $s_{\square}$ in $\coreset$ with weight $w(s_{\square})=\frac{1}{1-\eps'}\cdot\frac{g_\square}{M}\cdot n_{\square}$.
We say that all the points in $\Q_u(\I)\cap (\square\setminus B)$ are \emph{mapped} to $s_{\square}$.
We characterize $\square$ as a heavy cell and we add it to $B$. Otherwise, if $\frac{g_\square}{M}<2\tau$, then $\square$ is a light cell, we skip it and continue processing the next cell.
An example of the sampling procedure is shown in Figure~\ref{fig:sample}.
In the end, after repeating the algorithm for each $x_i\in \kapprox$, we have a weighted set $\coreset$ of size $O(|\kapprox|\eps^{-d_u}\log n)$. We run the standard algorithm for the weighted $k$-median problem $\kmedianAlg_\gamma$ (or discrete $k$-median problem $\DkmedianAlg_\gamma$) on $\coreset$ to get a set of $k$ centers $\ret$. We return the set of centers $\ret$. Furthermore, we return $r_u=\frac{1+4\eps'}{1-9\eps'}\kmedian_{\ret}(\coreset)$.

\paragraph{Correctness}
Let $P_u$ (heavy tuples) be the tuples in $\Q_u(\I)$ that belong to at least one heavy cell in $B$ and let $\bar{P}_u=\Q_u(\I)\setminus P_u$ (light tuples), at the end of the algorithm. By construction, every tuple $t\in \Q_u(\I)$, belongs to at least one heavy or light cell. Notice that a point $t\in \Q_u(\I)$ might first lie in a light cell and only later in the algorithm it might be found in a heavy cell.
It is straightforward to see that any heavy point $t\in P_u$ is mapped to a point in $\coreset$.
For each $t\in P_u$,
let $\square_t$ be the first heavy cell visited by the algorithm such that $t\in \Q_u(\I)\cap \square_t$.
Let $i(t)$ be the index such that $\square_t\in \widebar{V}_{i(t)}$.
In the deterministic algorithm, we had the assignment function $\sigma(\cdot)$ for all the points in $\Q_u(D)$. In this algorithm, we only map the points in $P_u$, so we define a new mapping function $\hat{\sigma}(\cdot)$, i.e., for a point $t\in P_u$, $\hat{\sigma}(t):=s_{\square_t}\in\coreset$. 
For a cell $\square$ processed by the algorithm, let $B_\square$ be the set of heavy cells found by the algorithm just before $\square$ was processed.
For a point $p\in \coreset$, let $\square_p$ be the cell that the algorithm processed while $p$ was added in $\coreset$, and let $n_{p}=|\Q_u(\I)\cap (\square_p\setminus B_{\square_p})|$ be the number of the points mapped to $p$.

We show the correctness of our algorithm through a number of technical lemmas.
We first show a crucial observation. There exists a \emph{charging process} where i) every light tuple (i.e., a tuple that does not belong to a heavy cell) charges $\frac{1}{\eps'}$ heavy tuples that lie to the same cell as the light tuple, and ii) every heavy tuple is charged at most once. This observation is then used to show that for any set of $k$ centers, the $k$-median error of $\Q_u(\I)$ is close enough to the $k$-median error with respect to the heavy tuples so we can safely ignore the other tuples. Using this argument, we prove that $\coreset$ is a $9\eps'$-coreset for the heavy points. Using all previous observations, we conclude that $\ret$ is a $(1+\eps)\gamma$-approximation for the relational $k$-median clustering. All missing lemmas and proofs can be found in Appendix~\ref{appndx:fast}.

\begin{lemma}
\label{lem:tech2}
    For every point $p\in \coreset$, $n_{p}\leq w(p)\leq (1+4\eps')n_{p}$ with probability at least $1-\frac{1}{N^{O(1)}}$.
\end{lemma}
\begin{proof}

    Let $p\in\coreset$ be a point in the coreset. By definition $\frac{g_{\square_p}}{M}\geq 2\tau$, so
    using the Chernoff bound, as shown  in~\cite{chen2022coresets} (Lemma 2),  with probability at least $1-\frac{1}{N^{O(1)}}$ it holds that $$(1-\eps')|\Q_u(\I)\cap (\square_p\setminus B_{\square_p})|\leq \frac{g_{\square_p}}{M}n_{\square_p}\leq (1+\eps')|\Q_u(\I)\cap (\square_p\setminus B_{\square_p})|.$$
    Hence, $n_{p}\leq w(p)\leq \frac{1+\eps'}{1-\eps'}n_{p}\leq (1+4\eps')n_{p}$ with probability at least $1-1/N^{O(1)}$.
\end{proof}

Let $L$ be the set of light cells found by the algorithm. For a cell $\square\in L$, let $\mathcal{P}^L_\square$ be the set of the points in $\Q_u(\I)\cap \square$ that do not belong in a heavy cell at the time that the algorithm processes $\square$. Notice that a point $p\in \mathcal{P}^L_\square$ might also belong to $P_u$ (points that lie in at least one heavy cell) because $p$ was found inside a heavy cell later in the algorithm.
In addition, we note that it might be possible that  $\mathcal{P}^L_{\square_j}\cap \mathcal{P}^L_{\square_h}\neq \emptyset$, for $j\neq h$.
Furthermore, let $\mathcal{P}^B_\square$ be the 
set of points in $\Q_u(\I)\cap \square$ that belong in at least one heavy cell at the time that the algorithm processes $\square$. 

In the next technical lemma, we show the existence of a charging process that is later used to estimate the $k$-median error of $\Q(\I)$ using only the heavy tuples $P_u$.
\begin{lemma}
    \label{lem:assignment}
   There exists a charging process that works with probability at least $1-\frac{1}{N^{O(1)}}$ having the following properties. For every cell $\square\in L$, each point $p\in \mathcal{P}_\square^L$ \emph{charges} $\frac{1}{\eps'}$ points in  $\mathcal{P}_\square^B$, such that, in the end, every point $t\in P_u$ has been charged at most once.
\end{lemma}
\begin{proof}
    Let $L=\{\square_1,\ldots, \square_\eta\}$ be sorted in ascending order of $|\mathcal{P}_\square^L|$, i.e., $|\mathcal{P}_{\square_j}^L|\leq |\mathcal{P}_{\square_{j+1}}^L|$.
    It is sufficient to show that for every $j=1,\ldots,\eta$, $|\mathcal{P}_{\square_j}^B|-\frac{1}{\eps'}\sum_{h< j}|\mathcal{P}_{\square_h}^L|\geq \frac{1}{\eps'}|\mathcal{P}_{\square_j}^L|.$
    Indeed, if this inequality holds for every $j$, then there are always at least $\frac{1}{\eps'}|\mathcal{P}_{\square_j}^L|$ uncharged points in $\mathcal{P}_{\square_j}^B$, and we can charge each point in $\mathcal{P}_{\square_j}^L$ to $\frac{1}{\eps'}$ points in $\mathcal{P}_{\square_j}^B$.

    Recall that $\square_j\in L$ is a light cell because the algorithm found that the ratio $\frac{|H_{\square_j}\cap \mathcal{P}_{\square_j}^L|}{M}\leq 2\tau$. In this case, from the Chernoff bound (Lemma 2 in the full version of~\cite{chen2022coresets}) we have $\frac{|\mathcal{P}_{\square_j}^L|}{|\mathcal{P}_{\square_j}^B|+|\mathcal{P}_{\square_j}^L|}\leq 4\tau$ with probability at least $1-\frac{1}{N^{O(1)}}$. Solving the inequality with respect to $|\mathcal{P}_{\square_j}^B|$, we get $|\mathcal{P}_{\square_j}^B|\geq \frac{1-4\tau}{4\tau}|\mathcal{P}_{\square_j}^L|\geq \frac{1}{8\tau}|\mathcal{P}_{\square_j}^L|$, because $\tau< \frac{1}{8}$. Hence, $|\mathcal{P}_{\square_j}^B|\geq 2|\kapprox|(\eps')^{-d_u-1}\log (N)\cdot |\mathcal{P}_{\square_j}^L|$.
    
    Next, we find an upper bound for $\frac{1}{\eps'}\sum_{h<j}|\mathcal{P}_{\square_h}^L|$. Because of sorting in ascending order of $|\mathcal{P}_{\square_j}^L|$, we have $\frac{1}{\eps'}\sum_{h<j}|\mathcal{P}_{\square_h}^L|\leq \frac{j}{\eps'}|\mathcal{P}_{\square_j}^L|\leq \frac{|L|}{\eps'}|\mathcal{P}_{\square_j}^L|\leq |\kapprox|(\eps')^{-d_u-1}\log(N)\cdot |\mathcal{P}_{\square_j}^L|.$
    
    Hence, we conclude that
   $$|\mathcal{P}_{\square_j}^B|-\frac{1}{\eps'}\sum_{h< j}|\mathcal{P}_{\square_j}^L|\geq 2|\kapprox|(\eps')^{-d_u-1}\log (N) |\mathcal{P}_{\square_j}^L|-|\kapprox|(\eps')^{-d_u-1}\log (N) |\mathcal{P}_{\square_j}^L|\geq \frac{1}{\eps'}|\mathcal{P}_{\square_j}^L|.$$
    \vspace{-1em}
\end{proof}
For each point $p\in \bar{P}_u$, let $\square_{i(p)}\in L$ be the cell in $L$ such that $p\in \square_{i(p)}$ and $p$ charges $\frac{1}{\eps'}$ points in $\mathcal{P}_{\square_{i(p)}}^B$, as shown in Lemma~\ref{lem:assignment}.
Let $t_{j_1(p)}, \ldots, t_{j_{1/\eps'}(p)}$ be these points in $\mathcal{P}_{\square_{i(p)}}^B$.

Next, we show that the $k$-median error with respect to the heavy points $P_u$ approximates the $k$-median error of all tuples in $\Q_u(\I)$.
\begin{lemma}
\label{lem:genLight}
    Let $Y$ be an arbitrary set of $k$ points in $\Re^{d_u}$. It holds that $\kmedian_Y(\bar{P}_u)\leq \eps'\kmedian_Y(P_u)+\eps'\kmedian_{Y}(\Q_u(\I))$ and 
    $\kmedian_Y(\Q_u(\I))\leq (1+4\eps')\kmedian_Y(P_u)$
    with probability at least $1-\frac{1}{N^{O(1)}}$.
\end{lemma}
\begin{proof}
    We start showing the first inequality $\kmedian_Y(\bar{P}_u)\leq \eps'\kmedian_Y(P_u)+\eps'\kmedian_{Y}(\Q_u(\I))$.
    We have $\kmedian_Y(\bar{P}_u)=\sum_{p\in \bar{P}_u}\dist(p,Y)$.
    For a point $p\in \bar{P}_u$  and each $h\leq \frac{1}{\eps'}$, by triangle inequality, we have
    $$\dist(p,Y)\leq \dist(t_{j_h(p)}, Y)+\dist(p, t_{j_h(p)}).$$
    Taking the sum of these $\frac{1}{\eps'}$ inequalities, we have
    $\dist(p, Y)\leq \eps'\sum_{h=1}^{1/\eps'}\dist(t_{j_h(p)},Y) +\eps' \sum_{h=1}^{1/\eps'}\dist(p,t_{j_h(p)}),$
    so we get 
    \vspace{-1em}
    $$\kmedian_Y(\bar{P}_u)\leq\eps' \sum_{p\in\bar{P}_u}\sum_{h=1}^{1/\eps'}\dist(t_{j_h(p)},Y) + \eps' \sum_{p\in\bar{P}_u}\sum_{h=1}^{1/\eps'}\dist(p,t_{j_h(p)}).$$
    From Lemma~\ref{lem:assignment}, we proved that any $t\in P_u$ is charged by at most one point in $\bar{P}_u$, so the first term in the sum can be bounded as
    $\eps' \sum_{p\in\bar{P}_u}\sum_{h=1}^{1/\eps'}\dist(t_{j_h(p)},Y)\leq \eps'\sum_{t\in P_u}\dist(t,Y)=\eps'\kmedian_Y(P_u).$
    In Appendix~\ref{appndx:fast} we bound the second term in the sum showing that $\eps' \sum_{p\in\bar{P}_u}\sum_{h=1}^{1/\eps'}\dist(p,t_{j_h(p)})\leq \eps'\kmedian_Y(\Q(\I))$.
    Hence, the first inequality follows.

    For the second inequality we have,
    $\kmedian_Y(P_u)=\kmedian_Y(\Q_u(\I))-\kmedian_Y(\bar{P}_u)\geq \kmedian_Y(\Q_u(\I))-\eps'\kmedian_Y(P_u)-\eps'\kmedian_Y(\Q_u(\I))$, so $\kmedian_Y(\Q_u(\I))\leq \frac{1+\eps'}{1-\eps'}\kmedian_Y(P_u)\leq (1+4\eps')\kmedian_Y(P_u)$.
\end{proof}
Using the inequalities in Lemma~\ref{lem:genLight} and Lemma~\ref{lem:tech2}, we follow the proof of Lemma~\ref{lem:coreset1} and we show the next result.
\begin{lemma}
\label{lem:coreset2}
    $\coreset$ is an $9\eps'$-coreset of $P_u$ with probability at least $1-\frac{1}{N^{O(1)}}$.
\end{lemma}
\begin{proof}
    Let $Y$ be an arbitrary set of $k$ points in $\Re^{d_u}$, as we had in Lemma~\ref{lem:coreset1}.
    From Lemma~\ref{lem:tech2}, let $\eps_p\in [0,4\eps']$ be a real number such that $w(p)=(1+\eps_p)n_p$, for $p\in \coreset$.
    With probability at least $1-\frac{1}{N^{O(1)}}$, we have,
    \begin{align*}
    \mathcal{E}&=|\kmedian_Y(P_u)-\kmedian_Y(\coreset)|=|\sum_{t\in P_u}\dist(t,Y)-\sum_{p\in \coreset}w(p)\dist(p,Y)|= |\sum_{t\in P_u}\dist(t,Y)-\sum_{p\in \coreset}(1+\eps_p)n_{p}\dist(p,Y)|\\&=|\sum_{t\in P_u}\dist(t,Y)-\sum_{t\in P_u}(1+\eps_{\hat{\sigma}(t)})\dist(\hat{\sigma}(t),Y)| 
    \leq \sum_{t\in P_u}|\dist(t,Y)-(1+\eps_{\hat{\sigma}(t)})\dist(\hat{\sigma}(t),Y)|.
    \end{align*}
    This inequality also follows from Lemma~\ref{lem:tech2}. Indeed, each $\hat{\sigma}(t)$ has a weight that is $(1+\eps_{\hat{\sigma}(t)})$ times larger than the number of points in $P_u$ that are assigned to $\hat{\sigma}(t)$.
We have,
    \begin{align*}
        \mathcal{E}&\leq \sum_{t\in P_u}|\dist(t,Y)-(1+\eps_{\hat{\sigma}(t)})\dist(\hat{\sigma}(t),Y)|\leq \sum_{t\in P_u}|\dist(t,Y)-\dist(\hat{\sigma}(t),Y)|+4\eps'\sum_{t\in P_u}\dist(\hat{\sigma}(t),Y)\\&\leq \sum_{t\in P_u}\dist(t,\hat{\sigma}(t))+4\eps'\kmedian_Y(P_u).
    \end{align*}
    Following the proof of Lemma~\ref{lem:coreset1} we have 
    $\sum_{t\in P_u}\dist(t,\hat{\sigma}(t))\leq \eps'\kmedian_Y(\Q_u(\I))$.
    From Lemma~\ref{lem:genLight} we get $\kmedian_Y(\Q_u(\I))\leq (1+4\eps')\kmedian_Y(P_u)$. We conclude that $\mathcal{E}\leq (5\eps'+4(\eps')^2)\kmedian_Y(P_u)\leq 9\eps'\kmedian_Y(P_u)$.
\end{proof}

Form Lemma~\ref{lem:coreset2}, we conclude to the main result.
\begin{lemma}
\label{lem:main2}
If $\kmedianAlg_\gamma$ is used, then $\ret\subset \Re^d$ and 
$\kmedian_{\ret}(\Q(\I))\leq r_u\leq (1+\eps)\gamma\kmedian_{\opt(\Q(\I))}(\Q(\I))$, with probability at least $1-\frac{1}{N^{O(1)}}$. If $\DkmedianAlg_\gamma$ is used, then 
$\ret\subseteq \Q_u(\I)$ and
$\kmedian_\ret(\Q_u(\I))\leq r_u\leq (2+\eps)\gamma \kmedian_{\optD(\Q_u(\I))}(\Q_u(\I))$, with probability at least $1-\frac{1}{N^{O(1)}}$.
\end{lemma}
\begin{proof}
    We first consider the case where $\kmedianAlg_\gamma$ is used.
   From Lemma~\ref{lem:coreset2}, we have that for any set of $k$ points $Y$ in $\Re^{d_u}$, $(1-9\eps')\kmedian_Y(P_u)\leq \kmedian_Y(\coreset)\leq (1+9\eps')\kmedian_Y(P_u)$, with probability at least $1-\frac{1}{N^{O(1)}}$.
    By definition, 
    \begin{equation}
        \label{eq:corfinal}
        \kmedian_{\ret}(\coreset)\leq \gamma\kmedian_{\opt(\coreset)}(\coreset)\leq \gamma\kmedian_{\opt(P_u)}(\coreset)\leq (1+9\eps')\gamma\kmedian_{\opt(P_u)}(P_u).
    \end{equation}
    The last inequality follows by the definition of the coreset for $Y=\opt(P_u)$.
    Since $P_u\subseteq \Q_u(\I)$ and $\opt(P_u), \opt(\Q_u(\I))\subset \Re^d$, it also holds that $\kmedian_{\opt(P_u)}(P_u)\leq \kmedian_{\opt(\Q_u(\I))}(\Q_u(\I))$.

    From Lemma~\ref{lem:genLight} (for $Y=\ret$) we have $\kmedian_{\ret}(\Q_u(\I))\leq (1+4\eps')\kmedian_{\ret}(P_u)$. Hence,
    
   \begin{align*}
       \kmedian_{\ret}(\Q_u(\I))&\leq (1+4\eps')\kmedian_{\ret}(P_u)\leq \frac{1+4\eps'}{1-9\eps'}\kmedian_{\ret}(\coreset)=r_u\leq \frac{(1+4\eps')(1+9\eps')}{1-9\eps'}\gamma\kmedian_{\opt(\Q_u(\I))}(\Q_u(\I))\\
       &\leq (1+34\eps')\gamma\kmedian_{\opt(\Q_u(\I))}(\Q_u(\I))=(1+\eps)\gamma\kmedian_{\opt(\Q_u(\I))}(\Q_u(\I)).
   \end{align*}
 The proof using the $\DkmedianAlg_\gamma$ algorithm is shown in Appendix~\ref{appndx:fast}.
\end{proof}

\mparagraph{Running time}
The total number of cells that the algorithm processes is $O(|\kapprox|\eps^{-d_u}\log N)$. In each cell, we take $M=O(|\kapprox|\eps^{-d_u-3}\log^2 N)$ samples. Using Lemma~\ref{lem:Rects}, for each cell we spend $O(N+M\log N)$ time to get the set of samples. For each sample, we check whether it belongs to a heavy cell in $B$. This can be checked trivially in $O(|B|)=O(|\kapprox|\eps^{-d_u}\log N)$ time, or in $O(\log^{d_u}(|\kapprox|\eps^{-d_u}\log N))$ time using a dynamic geometric data structure for stabbing queries~\cite{agarwal1999geometric, sun2019parallel}.
Finally, the standard algorithm for $k$-median clustering on $O(|\kapprox|\eps^{-d_u}\log N)$ points takes $O(\timeMedian_\gamma(|\kapprox|\eps^{-d_u}\log N))$ time.
The overall time of the algorithm
is
$$O\left(\!N|\kapprox|\eps^{-d_u}\log(N)\!+\! |\kapprox|^2\eps^{-2d_u-3}\log^3 (N)\min\left\{|\kapprox|\eps^{-d_u}\log (N), \log^{d_u}(|\kapprox|)\right\} \!+\!\timeMedian_\gamma(|\kapprox|\eps^{-d_u}\log N)\right).$$

\begin{theorem}
    \label{coreset2}
    Let $\I$ be a database instance with $N$ tuples, $\Q$ be an acyclic join query over a set of attributes $\allattr$, and $\allattr_u\subseteq \allattr$. Given a set $X\subset \Re^d$, a constant $\alpha$ such that $\kmedian_X(\Q_u(\I))\leq \alpha \kmedian_{\opt(\Q_u(\I))}(\Q_u(\I))$, and a constant parameter $\eps\in(0,1)$, there exists an algorithm that computes a set $\ret\subset \Re^d$ of $k$ points and a number $r_u$ in $O(|\kapprox|N\log N+|X|^2\log^3(N)\min\{\log^{d_u}(|X|),|X|\log N\}+\timeMedian_\gamma(|\kapprox|
\log N))$ such that $\kmedian_{\ret}(\Q_u(\I))\leq r_u\leq (1+\eps)\gamma\kmedian_{\opt(\Q_u(\I))}(\Q_u(\I))$, with probability at least $1-\frac{1}{N^{O(1)}}$. There also exists an algorithm that computes a set $\ret\subseteq \Q_u(\I)$ of $k$ points and a number $r_u$ with the same running time, such that $\kmedian_{\ret}(\Q_u(\I))\leq r_u\leq (2+\eps)\gamma\kmedian_{\optD(\Q_u(\I))}(\Q_u(\I))$, with probability at least $1-\frac{1}{N^{O(1)}}$.
\end{theorem}
\section{Efficient algorithms}
\label{sec:kmedianAlg}
We use the results from Section~\ref{sec:coresets} to describe a complete algorithm for the relational $k$-median clustering. In the previous section, we saw how we can get a $(1+\eps)\gamma$-approximation algorithm if a set $\kapprox$ along with a number $r$ such that $\kmedian_\kapprox(\Q(\I))\leq r\leq O(1)\kmedian_{\opt(\Q(\I))}\Q(\I)$ are given. Essentially, in this section we focus on computing such a set $\kapprox$ and number $r$, efficiently.
\subsection{Acyclic queries}
\label{subsec:algmedian}

We construct a balanced binary tree $\tree$ such that the $j$-th leaf node stores the $j$-th attribute $A_j\in \allattr$ (the order is arbitrary).
For a node $\node$ of $\tree$, let $\tree_\node$ be the subtree of $\tree$ rooted at $\node$.
Let $\allattr_\node$ be the set of attributes stored in the leaf nodes of $\tree_\node$. For every node $\node$, our algorithm computes a set $\ret_\node$ of cardinality $k$ and a real positive number $r_\node$ such that
\begin{equation}
    \label{eq:goal}
    \kmedian_{\ret_\node}(\Q_\node(\I))\leq r_\node\leq (1+\eps)\gamma\kmedian_{\opt(\Q_\node(\I))}(\Q_\node(\I)),
\end{equation}
for the relational $k$-median clustering. For the discrete relational $k$-median clustering, our algorithm computes a set $\ret_\node$ of cardinality $k$ and a real positive number $r_\node$ such that
\begin{equation}
    \label{eq:goal2}
    \kmedian_{\ret_\node}(\Q_\node(\I))\leq r_\node\leq (2+\eps)\gamma\kmedian_{\optD(\Q_\node(\I))}(\Q_\node(\I)),
\end{equation}

\begin{algorithm}[t]
\caption{{\sc Rel-k-Median}$(\Q, \I, \allrel, \eps, u)$}
\label{alg:FullAlg}
    \If{$u$ is a leaf node}{
        Let $R_j$ be a relation in $\allrel$ such that $A_u\in \allattr_j$\;
        Construct join tree $T$ of $\Q$ with $R_j$ in the root\;
        Run Yannakakis algorithm on $T$ to compute $c(h)=|\{t\in \Q(\I)\mid \pi_{\allattr_j}(t)=h\}|$, $\forall h\in R_j$\;
        $H_u=\pi_{A_u}(R_j)$\;
        \ForEach{$p\in H_u$}{$w(p)=\sum_{h\in R_j, \pi_{A_u}(h)=p}c(h)$\;
        }
        $\ret_u=\textsf{Opt\_1D\_kMedianAlg}((H_u,w(\cdot)))$;
        $\quad r_u=\kmedian_{\ret_u}(H_u)$\;
    }\Else{
        Let $v$ and $z$ be the children of $u$\;
        $\kapprox=\ret_v\times \ret_z$; 
        $\quad r=r_v+r_z$\;
        $(\ret_u, r_u)\leftarrow\textsf{RelClusteringFast}(\Q, \I, \allattr_u, \kapprox, (1+\eps)\gamma\sqrt{2}, r, \eps)$\;
    }
    \Return $(\ret_u, r_u)$\;
\end{algorithm}

The definition of $\Q_\node(\I)$ is the same as in the previous section: $\all=\widebar{\pi}_{\allattr_u}(\Q(\I)$). Similarly, we use the notation $d_\node=|\allattr_\node|$.
For a finite set $Y\subset \Re^d$ and a point $t\in \Re^d$ let $\nn(Y,t)=\argmin_{y\in Y}\dist(t,y)$, be the nearest neighbor of $t$ in $Y$.
Recall that for any set $Y\subset \Re^{d_\node}$,
$\kmedian_Y(\Q_\node(\I))=\sum_{t\in \Q(\I)}\dist(\pi_{\allattr_\node}(t),Y)=\sum_{t\in\Q(\I)}\sqrt{\sum_{A_j\in \allattr_\node}(\pi_{A_j}(t)-\pi_{A_j}(\nn(Y,\pi_{\allattr_\node}(t))))^2}.$

The algorithm we propose works bottom-up on tree $\tree$. If $\node$ has children $v, z$ we use the pre-computed $\ret_v, \ret_z$ and $r_v, r_z$ to compute $\ret_\node$ and $r_\node$.

\paragraph{Algorithm} First, we run the Yannakakis algorithm~\cite{yannakakis1981algorithms} on $\Q(\I)$ and we keep only the non-dangling tuples in $\I$, i.e., tuples in $\I$ that belong to at least one join result $\Q(\I)$. For any node $u$ of $\mathcal{T}$, in Algorithm~\ref{alg:FullAlg} we show the pseudocode for the geometric version of relational $k$-median clustering, computing $\ret_u$ and $r_u$.

Let $\node$ be a leaf node where $\allattr_\node$ contains one attribute $A_\node\in \allattr$. This is the only case that we try to implicitly construct $\Q_\node(\I)$ as a set of $O(N)$ points in $\Re^1$ with cardinalities (weights).
More specifically, we construct the weighted set of points in $\Re^1$, $H_\node=\pi_{A_\node}(\Q(\I))$ such that $p\in H_\node$ has weight $w(p)=|\{t\in \Q(\I)\mid \pi_{A_\node}(t)=p\}|$. 
Notice that $H_u$ is a set and not a multi-set. We can compute $H_u$ and the weight function $w(\cdot)$ as follows.
Let $R_j$ be an arbitrary relation from $\allrel$ that contains the attribute $A_\node$.
We use the counting version of Yannakakis algorithm to count the number of times that a tuple belongs in $\Q(\I)$. More specifically, we construct the join tree for $\Q$ and choose $R_j$ as its root. Using Yannakakis algorithm, for every tuple $h$ in the root $R_j$ we compute $c(h)=|\{t\in \Q(\I)\mid \pi_{\allattr_j}(t)=h\}|$.
By grouping together tuples from $R_j$ with the same value on attribute $A_\node$, we compute $H_u$ and the weight function $w(\cdot)$. More specifically, we set $H_u=\pi_{A_u}(R_j)$ and for each $p\in H_u$, we set $w(p)=\sum_{h\in R_j, \pi_{A_u}(h)=p}c(h)$.
Then, we execute the exact algorithm for the $1$-dimensional weighted $k$-median problem in the standard computational setting~\cite{gronlund2017fast, wu1991optimal, fleischer2006online} on set $H_u$ with weights $w(\cdot)$.
Let $\ret_\node$ be the returned set of $k$ centers. 
We also compute $r_\node=\kmedian_{\ret_\node}(H_u)$.

Next, assume that $\node$ is an inner node of $\tree$ with two children $v, z$. 
The algorithm sets $r=r_v+r_z$ and $\kapprox=\ret_v\times \ret_z$. Then, we run the algorithm from Theorem~\ref{coreset2} (or Theorem~\ref{coreset1}) using $\kapprox$ and $r$ as input  and compute $\ret_\node$ and $r_\node$.
More specifically, for the relational $k$-median clustering we call $\textsf{RelClusteringFast}(\Q, \I, \allattr_u, \kapprox, (1+\eps)\gamma\sqrt{2}, r, \eps)$
, while for the discrete relational $k$-median clustering we call $\textsf{RelClusteringFast}(\Q, \I, \allattr_u, \kapprox, 2(2+\eps)\gamma\sqrt{2}, r, \eps)$.
Let $\rootnode$ be the root of $\tree$. We return the set $\ret=\ret_\rootnode$. 


\paragraph{Correctness}
As we explained above, for the leaf nodes, the algorithm is simple. Let $\node$ be an intermediate node and let $v$ and $z$ be the two child nodes of $u$. Assuming that $\ret_v, r_v$ and $\ret_z, r_z$ satisfy the Equation~\eqref{eq:goal} (resp. Equation~\eqref{eq:goal2} for the discrete relational $k$-median), we show that $\ret_\node$ and $r_\node$ satisfy Equation~\eqref{eq:goal} (resp. Equation~\eqref{eq:goal2} for the discrete relational $k$-median).
If we prove that $\kmedian_{\kapprox}(\all)\leq r_u\leq \alpha \kmedian_{\opt(\all)}(\all)$, then the correctness follows from Theorem~\ref{coreset2}. The full proof of the next lemma can be found in Appendix~\ref{appndx:finalAlg}.

\begin{lemma}
    \label{lem:finalApprox}

    If $\kmedianAlg_\gamma$ is used, then $\kmedian_{\kapprox}(\all)\leq r_u\leq \alpha \kmedian_{\opt(\all)}(\all)$, for $\alpha=(1+\eps)\gamma\sqrt{2}$. If $\DkmedianAlg_\gamma$ is used, then $\kmedian_{\kapprox}(\all)\leq r_u\leq \alpha \kmedian_{\opt(\all)}(\all)$, for $\alpha=2(2+\eps)\gamma\sqrt{2}$.    
\end{lemma}
\begin{proof}
    We focus on the case where $\kmedianAlg_\gamma$ is used.
    Let $\mathcal{O}_v=\pi_{\allattr_v}(\opt(\all))$, and $\mathcal{O}_z=\pi_{\allattr_z}(\opt(\all))$. We define $\mathcal{O}=\mathcal{O}_v\times \mathcal{O}_z$. Notice that $\opt(\all)\subseteq \mathcal{O}$ so $\kmedian_{\mathcal{O}}(\all)\leq \kmedian_{\opt(\all)}(\all)$.
    We have,
    \begin{align*}
        \kmedian&_{\kapprox}(\all)=\sum_{t\in \Q(\I)}||\pi_{\allattr_\node}(t)-\nn(\kapprox,\pi_{\allattr_\node}(t))||\\&=\sum_{t\in\Q(\I)}\!\!\!\!\sqrt{||\pi_{\allattr_v}(t)-        
        \pi_{\allattr_v}(\nn(\kapprox,\pi_{\allattr_\node}(t)))||^2+||\pi_{\allattr_z}(t)-\pi_{\allattr_z}(\nn(\kapprox,\pi_{\allattr_\node}(t)))||^2}\\
        &\leq \sum_{t\in\Q(\I)}||\pi_{\allattr_v}(t)-\nn(\ret_v,\pi_{\allattr_v}(t))||+\sum_{t\in\Q(\I)}||\pi_{\allattr_z}(t)-\nn(\ret_z,\pi_{\allattr_z}(t))|| \leq r_v+r_z\\
        &\leq (1+\eps)\gamma\left(\sum_{t\in\Q(\I)}\!\!\!\!||\pi_{\allattr_v}(t)-\nn(\opt(\Q_v(\I)),\pi_{\allattr_v}(t))||\!+\!\!\!\sum_{t\in\Q(\I)}\!\!\!\!||\pi_{\allattr_z}(t)-\nn(\opt(\Q_z(\I)),\pi_{\allattr_z}(t))||\right)\\
        &\leq (1+\eps)\gamma\sqrt{2}\sum_{t\in \Q(\I)}\sqrt{||\pi_{\allattr_v}(t)-\pi_{\allattr_v}(\nn(\mathcal{O},\pi_{\allattr_\node}(t)))||^2+||\pi_{\allattr_z}(t)-\pi_{\allattr_z}(\nn(\mathcal{O},\pi_{\allattr_\node}(t)))||^2}\\
        &=(1+\eps)\gamma\sqrt{2}\sum_{t\in\Q(\I)}||\pi_{\allattr_\node}(t)-\nn(\mathcal{O},\pi_{\allattr_\node}(t))||\leq (1+\eps)\gamma\sqrt{2}\cdot \kmedian_{\opt(\all)}(\all).
    \end{align*}
    Similarly, if $\DkmedianAlg_\gamma$ is used, we have $\kmedian_{\kapprox}(\all)\!\leq\! r_u\!\leq\! 2(2+\eps)\gamma\sqrt{2}\cdot\kmedian_{\opt(\all)}(\all)$.
\end{proof}


\paragraph{Running time}
When $u$ is an inner node we spend $O(k^2)$ time to construct $X$ and then we run $\textsf{RelClusteringFast}(\Q, \I, \allattr_u, \kapprox, (1+\eps)\gamma\sqrt{2}, r, \eps)$, so the running time follows from Theorem~\ref{coreset2}.
When $u$ is a leaf node, all values $c(h)$ are computed in $O(N)$ time, with high probability. Hence $H_u$ and the weight function $w(\cdot)$ are computed in $O(N)$ time, with high probability. 
From~\cite{gronlund2017fast, wu1991optimal, fleischer2006online},
$\ret_\node$ is computed is $O(k\cdot N)$ time.

Putting everything together, we conclude with the next theorem.



\begin{theorem}
    \label{thm:mainRes}
    Given an acyclic join query $\Q$, a database instance $\I$, a parameter $k$, and a constant parameter $\eps\in(0,1)$, there exists an algorithm that computes a set $\ret\subset \Re^d$ of $k$ points such that $\kmedian_{\ret}(\Q(\I))\leq (1+\eps)\gamma\kmedian_{\opt(\Q(\I))}(\Q(\I))$, with probability at least $1-\frac{1}{N^{O(1)}}$. The running time of the algorithm is 
    $O(N k^2\log(N)+k^4\log^3(N)\min\{\log^{d}(k),k^2\log N\}+\timeMedian_\gamma(k^2\log N))$.
    Furthermore, there exists an algorithm that computes a set $\ret'\subseteq \Q(\I)$ of $k$ points in the same time such that $\kmedian_{\ret'}(\Q(\I))\leq (2+\eps)\gamma\kmedian_{\opt_{\textsf{disc}}(\Q(\I))}(\Q(\I))$, with probability at least $1-\frac{1}{N^{O(1)}}$.
\end{theorem}
We extend the result of Theorem~\ref{thm:mainRes} for the relational $k$-means clustering in Appendix~\ref{appndx:kmeans}.
Using Theorem~\ref{coreset1} instead of Theorem~\ref{coreset2} we can get a deterministic algorithm for the relational $k$-median clustering problem with the same approximation guarantees that runs in $O(k^{2d+2}N\log^{d+2} N+\timeMedian_\gamma(k^2\log N))$ time.

\subsection{Cyclic queries}
\label{sec:cycRes}
We use the notion of fractional hypertree width~\cite{gottlob2014treewidth} and use a standard procedure to extend our algorithms to every (cyclic) join query $\Q$.
For a cyclic join query $\Q$, we convert it to an equivalent acyclic query such that each relation is the result of a (possibly cyclic) join query with fractional edge cover at most $\fhw(\Q)$. We evaluate the (possibly cyclic) queries to derive the new relations and then apply the algorithm from Section~\ref{subsec:algmedian} on the new acyclic query.
Since it is a typical method in database theory, we give the details in Appendix~\ref{appndx:generalQueries}.
\begin{theorem}
    \label{thm:mainRes}
    Given a join query $\Q$, a database instance $\I$, a parameter $k$, and a constant parameter $\eps\in(0,1)$, there exists an algorithm that computes a set $\ret\subset \Re^d$ of $k$ points such that $\kmedian_{\ret}(\Q(\I))\leq (1+\eps)\gamma\kmedian_{\opt(\Q(\I))}(\Q(\I))$, with probability at least $1-\frac{1}{N^{O(1)}}$. The running time of the algorithm is $O(N^\fhw k^2\log(N)+k^4\log^3(N)\min\{\log^{d}(k),k^2\log N\}+\timeMedian_\gamma(k^2\log N))$.
    Furthermore, there exists an algorithm that computes a set $\ret'\subseteq \Q(\I)$ of $k$ points in the same time such that $\kmedian_{\ret'}(\Q(\I))\leq (2+\eps)\gamma\kmedian_{\opt_{\textsf{disc}}(\Q(\I))}(\Q(\I))$, with probability at least $1-\frac{1}{N^{O(1)}}$. The same results hold for the relational $k$-means clustering problem.
\end{theorem}
\vspace{-1em}
\section{Conclusion}
In this paper we propose improved approximation algorithms for the relational $k$-median and $k$-means clustering.
There are multiple interesting open problems derived from this work. It is interesting to check whether our geometric algorithms can be extended other clustering objective functions such as the sum of radii clustering. It is also interesting to use the ideas proposed in this paper to design clustering algorithms on the results of more complex queries such as conjunctive queries with inequalities or conjunctive queries with negation. Finally, someone can study relational clustering under different distance functions such as the Hamming, Jaccard or the cosine distance.




\bibliography{ref}
\newpage
\appendix
\section{Efficient algorithms for relational $k$-means clustering}
\label{appndx:kmeans}
\subsection{Coreset for $k$-means}
\label{sec:kmeans}

In this section, we focus on the relational $k$-means clustering. The algorithm is similar with the algorithm described in section \ref{sec:coresets}. However, there are changes in the analysis of the algorithms that need to be addressed. For simplicity, we focus on the (geometric) relational $k$-means clustering. All the results are straightforwardly extended to the discrete version using $\DkmeansAlg_\gamma$ instead of $\kmeansAlg_\gamma$ and using the fact that $\kmeans_{\opt(Y)}(Y)\leq \kmeans_{\optD(Y)}(Y)\leq 4\kmeans_{\opt(Y)}(Y)$, for any set of tuples $Y$, equivalently to the relational $k$-median clustering problem in Sections~\ref{sec:coresets},~\ref{sec:kmedianAlg}.

We keep the main notation as it is in section \ref{sec:coresets}. We change the definition of the set $\kapprox$ to be a set of points in $\Re^{d_u}$ such that $\kmeans_{\kapprox}(\Q_u(\I))\leq \alpha \cdot\kmeans_{\opt(\Q_u(\I))}(\Q_u(\I))$, where $\alpha>1$ is a constant.
Notice that $$\kmeans_{\kapprox}(\Q_u(\I))=\sum_{t\in \Q(\I)}\dist^2(\pi_{\allattr_u}(t), \kapprox).$$
We also assume that $r$ is a real number such that $\kmeans_{\kapprox}(\Q_u(\I))\leq r\leq \alpha \cdot\kmeans_{\opt(\Q_u(\I))}(\Q_u(\I))$.
In fact, we can also assume that $\frac{r}{1+\eps}\leq \kmeans_{\kapprox}(\Q_u(\I))$.
Note that again we do not assume anything about the size of $\kapprox$, and in the next section, we show that we can always consider $|\kapprox|=O(k^2)$. Again, we assume that at first, the set $\kapprox$ and the number $r$ are given as input, and later we describe the algorithm to efficiently construct $\kapprox$ and $r$. Similarly to the relational $k$-median clustering, we propose two algorithms, one slower deterministic and one faster randomized. The running times of the algorithms are the same as the running times of the algorithms for the relational $k$-median clustering.

\subsection{Deterministic algorithm for k-means}
\label{subsec:meansdet}

We set $\Phi = \sqrt{\frac{r}{\alpha n}}$ to be a lower bound estimate of the average mean radius $\sqrt{\frac{\kmeans_{\opt(\Q_u(\I))}(\Q_u(\I))}{n}}$, and keep all other parameters as described for the $k$-median coreset in section \ref{subsec:detmed}. Then we construct the same exponential grids and follow the exact same algorithms to get the coreset $\coreset$. Then, we run the standard algorithm for the weighted $k$-means problem on $\coreset$, $\kmeansAlg_\gamma(\coreset)$, to get a set of $k$ centers $\ret$. We return the set of centers $\ret$. Furthermore, we set and return $r_u=\frac{1}{1-\eps}\kmeans_{\ret}(\coreset)$.

\paragraph{Correctness}
We conduct a correctness analysis by proving analogous versions of the lemmas established for the $k$-median clustering, but this time for the $k$-means clustering.

For any tuple $t\in \Q_u(\I)$, let $x_i\in \kapprox$ be the center that is closest to $t$. We have that $\phi(t, x_i) \leq \alpha n \Phi$. So,
the following lemma is correct by the same argument as in lemma \ref{lem:tech1}.

\begin{lemma}
\label{lem:tech1kMeans}
    Every tuple $t\in \Q_u(\I)$ is assigned to one point in $\coreset$. Furthermore, the number of tuples in $\Q_u(\I)$ that are assigned to a point $s\in \coreset$ is $w(s)$.
\end{lemma}

Any point $t\in \Q_u(\I)$ is assigned to one point in $\coreset$, so we can define $x_{i(t)}$ and $\sigma(t)$ as before. 

\begin{lemma}
    \label{lem:coresetmeans}
    $\coreset$ is a $k$-means $\eps$-coreset for $\Q_u(\I)$.
\end{lemma}
\begin{proof}
Let $Y$ be an arbitrary set of $k$ points in $\Re^{d_u}$. The error is defined as
\begin{align*}
\error
=|\kmeans_Y(\Q(\I))-\kmeans_Y(\coreset)|
&\leq \sum_{t\in \Q_u(\I)}|\dist^2(t,Y)-\dist^2(\sigma(t),Y)| 
\\&= \sum_{t\in \Q_u(\I)}|(\dist(t,Y)-\dist(\sigma(t),Y))(\dist(t,Y)+\dist(\sigma(t),Y))|. 
\end{align*}
By the triangle inequality, $\dist(t,Y)\leq \dist(\sigma(t),Y)+\dist(t,\sigma(t))$ and $\dist(\sigma(t),Y)\leq \dist(t,Y)+\dist(t,\sigma(t))$. Hence, $|\dist(t,Y)-\dist(\sigma(t),Y)|\leq \dist(t,\sigma(t))$.
    We have, 
\begin{align*}
    \error 
    \leq \sum_{t\in \Q_u(\I)}|(\dist(t,Y)-\dist(\sigma(t),Y))(\dist(t,Y)+\dist(\sigma(t),Y))|
    \leq
    \sum_{t\in \Q_u(\I)}\dist(t, \sigma(t))(2\dist(t,Y)+\dist(t, \sigma(t)).
\end{align*}

Now we divide the points of $\Q_u(\I)$ into three cases. Let $P_R = \{t \in \Q_u(\I)| \dist(t,x_{i(t)})\leq \Phi \land \dist(t,Y)\leq \Phi\}$, and $P_X = \{t \in \Q_u(\I)\setminus P_R| \dist(t,Y)\leq \dist(t,x_{i(t)})\}$, and $P_Y = \{t \in \Q_u(\I)\setminus P_R| \dist(t,x_{i(t)})< \dist(t,Y)\}$. Note that we do not actually construct these sets, instead we only analyze the error induced by these three types of points separately. It is straightforward to see that $P_R\cup P_X\cup P_Y=\Q_u(\I)$ and $P_R\cap P_X=\emptyset$, $P_R\cap P_Y=\emptyset$, $P_X\cap P_Y=\emptyset$.

For a point $t\in P_R$, since $\dist(t,x_{i(t)})\leq \Phi$, by the construction of the exponential grid, we have that $\dist(t,\sigma(t))\leq \diameter(\square_t)\leq\frac{\eps}{10\alpha}\Phi$, hence 
\begin{align*}
    \sum_{t\in P_R}
    \dist(t, \sigma(t))(2\dist(t,Y)+\dist(t, \sigma(t))
    \leq 
    \sum_{t\in P_R}
    \frac{\eps}{10\alpha}\Phi (2\Phi + \frac{\eps}{10\alpha}\Phi)
    \leq
    \frac{3\eps}{10\alpha}n\Phi^2
    \leq
    \frac{3\eps}{10}\kmeans_{OPT(\Q_u(\I))}(\Q_u(\I)).
\end{align*}

By the definition of the set $P_X$, we have $\dist(t,x_{i(t)})>\Phi$, hence, as shown in the proof of lemma \ref{lem:coreset1},  $\dist(t,\sigma(t))\leq \frac{\eps}{10\alpha}\dist(t,x_{i(t)})$, and $\dist(t,x_{i(t)})\leq (1+\frac{4\eps}{10\alpha})\dist(t,\kapprox)$. Therefore, 

\begin{align*}
    &\sum_{t\in P_X}
    \dist(t, \sigma(t))(2\dist(t,Y)+\dist(t, \sigma(t))
    \leq 
    \sum_{t\in P_X}
    \frac{\eps}{10\alpha}\dist(t,x_{i(t)}) (2\dist(t,x_{i(t)}) + \frac{\eps}{10\alpha}\dist(t,x_{i(t)}))
    \\&\leq
    \frac{3\eps}{10\alpha} \sum_{t\in P_X}\dist^2(t,x_{i(t)})
    \leq 
    \frac{3\eps}{10\alpha}\sum_{t\in P_X}((1+\frac{4\eps}{10\alpha})\dist^2(t,\kapprox))
    \leq
    \frac{12\eps}{10\alpha}\kmeans_{\kapprox}(\Q_u(\I))
    \leq
    \frac{12\eps}{10}\kmeans_{OPT(\Q_u(\I))}(\Q_u(\I)).
\end{align*}

For the last case we have $\dist(t,\sigma(t))\leq \frac{\eps}{10\alpha}\dist(t,Y)$.
It holds because: If $\dist(t,x_{i(t)})\leq \Phi$, then $\dist(t,\sigma(t))\leq \frac{\eps}{10\alpha}\Phi\leq \frac{\eps}{10\alpha}\dist(t,Y)$. If $\dist(t,x_{i(t)})> \Phi$ then $\dist(t,\sigma(t))\leq \frac{\eps}{10\alpha}\dist(t,x_{i(t)})\leq \frac{\eps}{10\alpha}\dist(t,Y)$.
Hence, we have,
\begin{align*}
    \sum_{t\in P_Y}
    \dist(t, \sigma(t))(2\dist(t,Y)+\dist(t, \sigma(t))
    &\leq
    \sum_{t\in P_Y}\frac{\eps}{10\alpha}\dist(t,Y) (2\dist(t,Y) + \frac{\eps}{10\alpha}\dist(t,Y))
    \leq 
    \frac{3\eps}{10\alpha}\sum_{t\in P_Y} \dist^2(t, Y)\\&
    \leq \frac{3\eps}{10\alpha} \kmeans_{Y}(\Q_u(\I)).
\end{align*}

Finally, we bound the error,

\begin{align*}
    \error&\leq \sum_{t\in \Q_u(\I)}\dist(t, \sigma(t))(2\dist(t,Y)+\dist(t, \sigma(t))
    \\& \leq \sum_{t\in P_R}\dist(t, \sigma(t))(2\dist(t,Y)+\dist(t, \sigma(t))
    +
    \sum_{t\in P_X}\dist(t, \sigma(t))(2\dist(t,Y)+\dist(t, \sigma(t))
    \\&\hspace{15.5em}+
    \sum_{t\in P_Y}\dist(t, \sigma(t))(2\dist(t,Y)+\dist(t, \sigma(t))
    \\& \leq
    \frac{3\eps}{10}\kmeans_{OPT(\Q_u(\I))}(\Q_u(\I))
    + 
    \frac{12\eps}{10}\kmeans_{OPT(\Q_u(\I))}(\Q_u(\I))
    +
    \frac{3\eps}{10\alpha} \kmeans_{Y}(\Q_u(\I))
    \leq
    \frac{18\eps}{10} \kmeans_{Y}(\Q_u(\I)).
\end{align*}

Thus, if we set $\eps\leftarrow \eps/18$, the result follows.
\end{proof}

With the same argument as in lemma \ref{lem:correct1}, we can prove the following lemma.

\begin{lemma}
\label{lem:correctmeans}
    $\kmeans_\ret(\Q_u(\I))\leq r_u\leq (1+\eps)\gamma \kmeans_{\opt(\Q_u(\I))}(\Q_u(\I))$.
\end{lemma}

As the algorithm for constructing the $k$-means coreset is almost the same as the one for the $k$-median's coreset, and the only difference is in the analysis, the running time is the same and we can conclude with the following theorem.

\begin{theorem}
    \label{coreset:means}
 Let $\I$ be a database instance with $N$ tuples, $\Q$ be an acyclic join query over a set of attributes $\allattr$ and $\allattr_u\subseteq \allattr$. Given a set $X\subset \Re^d$, a constant $\alpha$ such that $\kmeans_X(\Q_u(\I))\leq \alpha \kmeans_{\opt(\Q_u(\I))}(\Q_u(\I))$, and a constant parameter $\eps\in(0,1)$, there exists an algorithm that computes a set $\ret\subset \Re^d$ of $k$ points and a number $r_u$ in $O(|\kapprox|^{d_u+1}N\log^{d_u+2} N+\timeMeans_\gamma(|\kapprox|
\log N))$ time such that $\kmeans_{\ret}(\Q_u(\I))\leq r_u\!\leq\! (1+\eps)\gamma\kmeans_{\opt(\Q_u(\I))}(\Q_u(\I))$. There also exists an algorithm that computes a set $\ret\subseteq \Q_u(\I)$ of $k$ points and a number $r_u$ with the same running time, such that $\kmeans_{\ret}(\Q_u(\I))\leq r_u\!\leq\! (4+\eps)\gamma\kmeans_{\optD(\Q_u(\I))}(\Q_u(\I))$.
\end{theorem}
\subsection{Randomized Algorithm for k-means Clustering}
\label{subsec:meansrand}
The algorithm for the $k$-means problem is exactly the same as the one we described in section \ref{subsec:medrand} for the $k$-median problem. Hence, in this section, we use the same notation and prove the additional lemmas required to show that the algorithm works also for $k$-means clustering. Note that this algorithm builds the $k$-means version of the exponential grid as described in \ref{subsec:meansdet}, and then continues with the randomized counting algorithm described in section \ref{subsec:medrand}.

\paragraph{Correctness}
Notice that Lemmas \ref{lem:tech2} and \ref{lem:assignment} are independent of the objective function of the problem and we can directly assume they are true without proving them again in this section.

We prove the following lemma using the same notation as in Lemma~\ref{lem:genLight}.

\begin{lemma}
\label{lem:genLightmeans}
    Let $Y$ be an arbitrary set of $k$ points in $\Re^{d_u}$. It holds that $\kmeans_Y(\bar{P}_u)\leq \eps\kmeans_Y(P_u)+\eps\kmeans_{Y}(\Q(\I))$ and $\kmeans_Y(\Q_u(\I))\leq (1+\eps)\kmeans_Y(P_u)$ with probability at least $1-\frac{1}{N^{O(1)}}$.
\end{lemma}
\begin{proof}
We first start showing the first inequality.
    For a point $p\in \bar{P}_u$ and each $h\leq \frac{1}{\eps}$, by triangle inequality, we have
    $$\dist(p,Y)\leq \dist(t_{j_h(p)}, Y)+\dist(p, t_{j_h(p)}).$$
    By squaring the above inequality we get,
     $$\dist^2(p,Y)\leq 2\dist^2(t_{j_h(p)}, Y)+2\dist^2(p, t_{j_h(p)}).$$
    Taking the sum of these $\frac{1}{\eps}$ inequalities, we have,
    $$\dist^2(p, Y)\leq 2\eps\sum_{h=1}^{1/\eps}\dist^2(t_{j_h(p)},Y) +2\eps \sum_{h=1}^{1/\eps}\dist^2(p,t_{j_h(p)}).$$
    So, we have $$\kmeans_Y(\bar{P}_u)\leq 2\eps \sum_{p\in\bar{P}_u}\sum_{h=1}^{1/\eps}\dist^2(t_{j_h(p)},Y) + 2\eps \sum_{p\in\bar{P}_u}\sum_{h=1}^{1/\eps}\dist^2(p,t_{j_h(p)}).$$
    From Lemma~\ref{lem:assignment}, we proved that any $t\in P_u$ is assigned by at most one point in $\bar{P}_u$, so the first term in the sum can be bounded as
    $$2\eps \sum_{p\in\bar{P}_u}\sum_{h=1}^{1/\eps}\dist^2(t_{j_h(p)},Y)\leq 2\eps\sum_{t\in P_u}\dist^2(t,Y)=2\eps\kmeans_Y(P_u).$$
    Next, we bound the second term in the sum.
    
    We define a new assignment function $\sigma'$ as we did in the proof of Lemma~\ref{lem:genLight}, and with the same argument we get,
    
    \begin{align*}
    2\eps \sum_{p\in\bar{P}_u}\sum_{h=1}^{1/\eps}\dist^2(p,t_{j_h(p)})
    \leq 
    2\eps\sum_{t\in P_u}\dist^2(t,\sigma'(t))
    \leq 
    2\eps \sum_{t\in \Q_u(\I)}\dist^2(t,\sigma'(t))
    \leq
    \frac{36\epsilon}{10}\kmeans_{Y}(\Q_{u}(\I)).
    \end{align*}
    The last inequality holds using the arguments in the proof of Lemma~\ref{lem:coresetmeans}.
    Finally, by summing up the two terms and setting $\eps\leftarrow \frac{\eps}{6}$, the first inequality follows.

    Then we prove the second inequality.
$$\kmeans_Y(P_u)=\kmeans_Y(\Q_u(\I))-\kmeans_Y(\bar{P}_u)\geq \kmeans_Y(\Q_u(\I))-\eps\kmeans_Y(P_u)-\eps\kmeans_Y(\Q_u(\I))\Leftrightarrow \kmeans_Y(\Q_u(\I))\leq \frac{1+\eps}{1-\eps}\kmeans_Y(P_u).$$
    If we set $\eps\leftarrow\eps/3$, the result follows.
\end{proof}

\begin{lemma}
\label{lem:coresetmeans2}
    $\coreset$ is a $k$-means $\eps$-coreset of $P_u$ with probability at least $1-\frac{1}{N^{O(1)}}$.
\end{lemma}
\begin{proof}
    Let $Y$ be an arbitrary set of $k$ points in $\Re^{d_u}$.
    With probability at least $1-\frac{1}{N^{O(1)}}$, we have,
    \begin{align*}
    \mathcal{E}&\!=\!|\kmeans_Y(P_u)-\kmeans_Y(\coreset)|\!=\!|\sum_{t\in P_u}\dist^2(t,Y)-\sum_{p\in \coreset}w(p)\dist^2(p,Y)|\!=\! |\sum_{t\in P_u}\dist^2(t,Y)-\sum_{p\in \coreset}(1+\eps_p)n_{p}\dist^2(p,Y)|\\&=|\sum_{t\in P_u}\dist^2(t,Y)-\sum_{t\in P_u}(1+\eps_{\hat{\sigma}(t)})\dist^2(\hat{\sigma}(t),Y)| 
    \leq \sum_{t\in P_u}|\dist^2(t,Y)-(1+\eps_{\hat{\sigma}(t)})\dist^2(\hat{\sigma}(t),Y)|.
    \end{align*}
    This inequality also follows from Lemma~\ref{lem:tech2}. We have,
    \begin{align*}
        \mathcal{E}&\leq \sum_{t\in P_u}|\dist^2(t,Y)-(1+\eps_{\hat{\sigma}(t)})\dist^2(\hat{\sigma}(t),Y)|
        \leq \sum_{t\in P_u}|\dist^2(t,Y)-\dist^2(\hat{\sigma}(t),Y)|+\eps\sum_{t\in P_u}\dist^2(\hat{\sigma}(t),Y)
        \\&\leq 
        \sum_{t\in P_u}|\dist^2(t,Y)-\dist^2(\hat{\sigma}(t),Y)|+2\eps\sum_{t\in P_u}\left(\dist^2(t, \hat{\sigma}(t)) + \dist^2(t,Y)\right).
    \end{align*}
    Using the proofs of Lemma \ref{lem:coresetmeans} and Lemma~\ref{lem:genLightmeans}, we have $\sum_{t\in P_u}|\dist^2(t,Y)-\dist^2(\hat{\sigma}(t),Y)| \leq (1+\eps)\frac{18\epsilon}{10}\kmeans_{Y}(P_u)$, and $\sum_{t\in P_u}\dist^2(t, \hat{\sigma}(t)) \leq (1+\eps)\frac{18\epsilon}{10}\kmeans_{Y}(P_u)$. Thus, we have,
    \begin{align*}
        \mathcal{E}&\leq 
        (1+\eps)\frac{18\epsilon}{10}\kmeans_{Y}(P_u) + (1+\eps)\frac{36\epsilon^2}{10}\kmeans_{Y}(P_u) + 2\epsilon\kmeans_{Y}(P_u) \leq \frac{130}{10}\epsilon\kmeans_{Y}(P_u).
    \end{align*}
    So, if we set $\epsilon \leftarrow \frac{\epsilon}{13}$, the result follows.
\end{proof}

Finally, we show that $\kmeans_{\ret}(\Q_u(\I))$ is a good approximation of $\kmeans_{\opt(\Q_u(\I)}(\Q_u(\I))$ and that $r_u = \frac{1+\epsilon}{(1-\epsilon)^2}\kmeans_\ret(\coreset)$ is a good estimate of $\kmeans_{\ret}(\Q_u(\I))$.

\begin{lemma}
\label{lem:main2Means}
    $\kmeans_{\ret}(\Q_u(\I))\leq r_u\leq (1+\eps)\gamma\kmeans_{\opt(\Q_u(\I))}(\Q_u(\I))$, with probability at least $1-\frac{1}{N^{O(1)}}$.
\end{lemma}
\begin{proof}
   From Lemma~\ref{lem:coresetmeans}, we have that for any set of $k$ points $Y$ in $\Re^{d_u}$, $(1-\eps)\kmeans_Y(P_u)\leq \kmeans_Y(\coreset)\leq (1+\eps)\kmeans_Y(P_u)$, with probability at least $1-\frac{1}{N^{O(1)}}$.
    By definition, 
    \begin{equation}
        \label{eq:corfinal}
        \kmeans_{\ret}(\coreset)\leq \gamma\kmeans_{\opt(\coreset)}(\coreset)\leq \gamma\kmeans_{\opt(P_u)}(\coreset)\leq (1+\eps)\gamma\kmeans_{\opt(P_u)}(P_u).
    \end{equation}
    The last inequality follows by the definition of the coreset for $Y=\opt(P_u)$.
    Since $P_u\subseteq \Q_u(\I)$ and $\opt(P_u), \opt(\Q_u(\I))\subset \Re^d$, it also holds that $\kmeans_{\opt(P_u)}(P_u)\leq \kmeans_{\opt(\Q_u(\I))}(\Q_u(\I))$.

    From Lemma~\ref{lem:genLightmeans} (for $Y=\ret$) we have $\kmeans_{\ret}(\Q_u(\I))\leq (1+\eps)\kmeans_{\ret}(P_u)$. Hence,
   \begin{align*}
       \kmeans_{\ret}(\Q_u(\I))&\leq (1+\eps)\kmeans_{\ret}(P_u)\leq \frac{1+\eps}{1-\eps}\kmeans_{\ret}(\coreset)=r_u\leq \frac{(1+\eps)^2}{1-\eps}\gamma\kmeans_{\opt(\Q_u(\I))}(\Q_u(\I)).
   \end{align*}
   If we set $\eps\leftarrow \eps/5$, the result follows.
\end{proof}

Notice that there is no difference between the above algorithm and the one for $k$-median proposed in section \ref{subsec:medrand}, and the only difference is in the correctness analysis. Therefore, the running is exactly the same, and we can directly have the following theorem.

\begin{theorem}
    \label{meansRandTheorem}
Let $\I$ be a database instance with $N$ tuples, $\Q$ be an acyclic join query over a set of attributes $\allattr$, and $\allattr_u\subseteq \allattr$. Given a set $X\subset \Re^d$, a constant $\alpha$ such that $\kmeans_X(\Q_u(\I))\leq \alpha \kmeans_{\opt(\Q_u(\I))}(\Q_u(\I))$, and a constant parameter $\eps\in(0,1)$, there exists an algorithm that computes a set $\ret\subset \Re^d$ of $k$ points and a number $r_u$ in $O(|\kapprox|N\log N+|X|^2\log^3(N)\min\{\log^{d_u}(|X|),|X|\log N\}+\timeMeans_\gamma(|\kapprox|
\log N))$ such that $\kmeans_{\ret}(\Q_u(\I))\leq r_u\leq (1+\eps)\gamma\kmeans_{\opt(\Q_u(\I))}(\Q_u(\I))$, with probability at least $1-\frac{1}{N^{O(1)}}$. There also exists an algorithm that computes a set $\ret\subseteq \Q_u(\I)$ of $k$ points and a number $r_u$ with the same running time, such that $\kmeans_{\ret}(\Q_u(\I))\leq r_u\leq (4+\eps)\gamma\kmeans_{\optD(\Q_u(\I))}(\Q_u(\I))$, with probability at least $1-\frac{1}{N^{O(1)}}$.

\end{theorem}

\subsection{k-means Clustering}
\label{subsec:algmeans}
We use the results from the previous section to describe a complete algorithm for the relational $k$-means clustering. As we did in section \ref{subsec:algmedian}, in this section, we describe how to efficiently construct the set $\kapprox$ and the number $r$. To do so, we construct a binary tree as described in section \ref{subsec:algmedian}, and use the same notation. Our goal this time is to compute for each node $\node$, a set $\ret_\node$, and a number $r_\node$, such that,

\begin{equation}
    \label{eq:goalmeans}
    \kmeans_{\ret_\node}(\Q_\node(\I))\leq r_\node\leq (1+\eps)\gamma\kmeans_{\opt(\Q_\node(\I))}(\Q_\node(\I)).
\end{equation}

The algorithm is exactly the same as the algorithm for $k$-median clustering. However, instead of running a standard $k$-median algorithm on the leaf node, we run the exact $O(k\cdot N)$ time algorithm for the $1$-dimensional weighted $k$-means problem in the standard computational setting~\cite{gronlund2017fast, wu1991optimal, fleischer2006online}, and on the intermediate nodes, we use the $k$-means version of our coresets described in Section~\ref{sec:kmeans}.
At last, we report $\ret = \ret_\rootnode$, where $\rootnode$ is the root of the tree.

\paragraph{Correctness}

Recall that for any set $Y\subseteq \Re^{d_\node}$,

$$\kmeans_Y(\Q_\node(\I))=\sum_{t\in \Q(\I)}\dist^2(\pi_{\allattr_\node}(t),Y)=\sum_{t\in\Q(\I)}\sum_{A_j\in \allattr_\node}(\pi_{A_j}(t)-\pi_{A_j}(\nn(Y,\pi_{\allattr_\node}(t)))^2.$$

Assuming that $\ret_v, r_v$ and  $\ret_z, r_z$ satisfy Equation~\ref{eq:goalmeans}, we show how to compute $\ret_\node$ and $r_\node$ that satisfy Equation~\ref{eq:goalmeans}.

If we prove that $\kapprox$ is a constant $\alpha$-approximation of the $k$-means problem on $\all$ and $\kmeans_{\kapprox}(\all)\leq r_u\leq \alpha \kmeans_{\opt(\all)}(\all)$, then the correctness follows from Theorem~\ref{meansRandTheorem} (or Theorem~\ref{coreset:means}).

\begin{lemma}
    \label{lem:finalApproxMeans}
    $\kmeans_{\kapprox}(\all)\leq r_u\leq \alpha \kmeans_{\opt(\all)}(\all)$.
\end{lemma}
\begin{proof}
    Let $\mathcal{O}_v=\pi_{\allattr_v}(\opt(\all))$, and $\mathcal{O}_z=\pi_{\allattr_z}(\opt(\all))$. We define $\mathcal{O}=\mathcal{O}_v\times \mathcal{O}_z$. Notice that $\opt(\all)\subseteq \mathcal{O}$ so $\kmeans_{\mathcal{O}}(\all)\leq \kmeans_{\opt(\all)}(\all)$.
    We have,
    \begin{align*}
        \kmeans_{\kapprox}&(\all)=\sum_{t\in \Q(\I)}||\pi_{\allattr_\node}(t)-\nn(\kapprox,\pi_{\allattr_\node}(t))||^2\\
        &=\sum_{t\in\Q(\I)}||\pi_{\allattr_v}(t)-\pi_{\allattr_v}(\nn(\kapprox,\pi_{\allattr_\node}(t)))||^2+\sum_{t\in\Q(\I)}||\pi_{\allattr_z}(t)-\pi_{\allattr_z}(\nn(\kapprox,\pi_{\allattr_\node}(t)))||^2\\
        &= \sum_{t\in\Q(\I)}||\pi_{\allattr_v}(t)-\nn(\ret_v,\pi_{\allattr_v}(t))||^2+\sum_{t\in\Q(\I)}||\pi_{\allattr_z}(t)-\nn(\ret_z,\pi_{\allattr_z}(t))||^2\\
        &\leq r_v+r_z\\
        &\leq (1+\eps)\gamma\left(\sum_{t\in\Q(\I)}\!\!\!\!||\pi_{\allattr_v}(t)\!\!-\!\!\nn(\opt(\Q_v(\I)),\pi_{\allattr_v}(t))||^2
        \!\!+\!\!\!\!\!\sum_{t\in\Q(\I)}\!\!\!\!||\pi_{\allattr_z}(t)\!-\!\nn(\opt(\Q_z(\I)),\pi_{\allattr_z}(t))||^2\right)\\
        &\leq (1+\eps)\gamma\sum_{t\in\Q(\I)}||\pi_{\allattr_v}(t)-\nn(\mathcal{O}_v,\pi_{\allattr_v}(t))||^2+(1+\eps)\gamma\sum_{t\in\Q(\I)}||\pi_{\allattr_z}(t)-\nn(\mathcal{O}_z,\pi_{\allattr_z}(t))||^2\\
        &= (1+\eps)\gamma\sum_{t\in\Q(\I)}\left(||\pi_{\allattr_v}(t)-\pi_{\allattr_v}(\nn(\mathcal{O},\pi_{\allattr_\node}(t)))||^2+||\pi_{\allattr_z}(t)-\pi_{\allattr_z}(\nn(\mathcal{O},\pi_{\allattr_\node}(t)))||^2\right)\\
        &=(1+\eps)\gamma\sum_{t\in\Q(\I)}||\pi_{\allattr_\node}(t)-\nn(\mathcal{O},\pi_{\allattr_\node}(t))||^2=(1+\eps)\gamma\cdot\kmeans_{\mathcal{O}}(\all)\\
        &\leq (1+\eps)\gamma\cdot \kmeans_{\opt(\all)}(\all).
    \end{align*}
    The above inequalities are mostly simpler versions of the inequalities in Lemma~\ref{lem:finalApprox}, and they hold by the same arguments.
    Overall, we can set $\alpha=(1+\eps)\gamma$ which is a constant, and the result follows.
\end{proof}

We have that $\ret=\ret_\rootnode$ is a $(1+\eps)\gamma$ approximation of the relational $k$-median problem in $\Q(\I)$.

The algorithm is the same as the proposed algorithm for the $k$-median problem, and it has the same time complexity. Therefore we can conclude with the following theorem.

\begin{theorem}
    \label{thm:mainResmeans}
Given an acyclic join query $\Q$, a database instance $\I$, a parameter $k$, and a constant parameter $\eps\in(0,1)$, there exists an algorithm that computes a set $\ret\subset \Re^d$ of $k$ points such that $\kmeans_{\ret}(\Q(\I))\leq (1+\eps)\gamma\kmeans_{\opt(\Q(\I))}(\Q(\I))$, with probability at least $1-\frac{1}{N^{O(1)}}$. The running time of the algorithm is 
    $O(N k^2\log(N)+k^4\log^3(N)\min\{\log^{d}(k),k^2\log N\}+\timeMeans_\gamma(k^2\log N))$.
    Furthermore, there exists an algorithm that computes a set $\ret'\subseteq \Q(\I)$ of $k$ points in the same time such that $\kmeans_{\ret'}(\Q(\I))\leq (4+\eps)\gamma\kmeans_{\opt_{\textsf{disc}}(\Q(\I))}(\Q(\I))$, with probability at least $1-\frac{1}{N^{O(1)}}$.


\end{theorem}

Using Theorem~\ref{coreset:means}, we can also get a deterministic algorithm with the same guarantees that runs in $O(k^{2d+2}N\log^{d+2} N+\timeMeans_\gamma(k^2\log N))$ time.

\section{Missing Proofs from Subsection~\ref{subsec:detmed}}
\label{appndx:slowAlg}

\begin{proof}[Proof of Lemma~\ref{lem:correct1}]

    Next, we assume that  $\DkmedianAlg_\gamma$ is used. We note that \begin{align*}
    \kmedian_\ret(\Q_u(\I))&\leq \frac{1}{1-\eps'}\kmedian_{\ret}(\coreset)=r_u\leq 
    \gamma\frac{1}{1-\eps'}\kmedian_{\optD(\coreset)}(\coreset) 
    \leq 
    2\gamma\frac{1}{1-\eps'}\kmedian_{\opt(\coreset)}(\coreset)
     \\&\leq 2\frac{1+\eps'}{1-\eps'}\gamma \kmedian_{\opt(\Q_u(\I))}(\Q_u(\I))\leq 2(1+4\eps')\gamma \kmedian_{\opt(\Q_u(\I))}(\Q_u(\I))\\&=2(1+\eps)\gamma \kmedian_{\opt(\Q_u(\I))}(\Q_u(\I))\leq 2(1+\eps)\gamma \kmedian_{\optD(\Q_u(\I))}(\Q_u(\I)).
    \end{align*}
    The result follows setting $\eps\leftarrow \eps/2$.
\end{proof}

\section{Missing Proofs From Subsection~\ref{subsec:medrand}}
\label{appndx:fast}

\begin{proof}[Proof of Lemma~\ref{lem:genLight}]


    We bound the second term in the sum.
    We define a new assignment function $\sigma'$ as we did in Lemma~\ref{lem:coreset1}. 
    For a point $p\in \bar{P}_u$ we set $\sigma'(p)=p$. If a point $t\in P_u$, is not charged by any point in $\bigcup_{\square\in L}\mathcal{P}_{\square}^L$, then $\sigma'(t)=t$.
    For a point $t\in P_u$ let $p_t$ be the point in $\bigcup_{\square\in L}\mathcal{P}_{\square}^L$ such that $p_t$ charges $t=t_{j_h(p_t)}$ for a value of $h$, according to Lemma~\ref{lem:assignment}. If $p_t\in \bar{P}_u$ then  $\sigma'(t)=p_t$. If $p_t\in P_u$, $\sigma'(t)=t$.

    Using the new assignment function $\sigma'(\cdot)$, the second term can be bounded as 
    $$\eps' \sum_{p\in\bar{P}_u}\sum_{h=1}^{1/\eps'}\dist(p,t_{j_h(p)})\leq \eps'\sum_{t\in P_u}\dist(t,\sigma'(t))\leq \eps' \sum_{t\in \Q_u(\I)}\dist(t,\sigma'(t)).$$
    We note that $\sigma'$ is a different assignment than the assignment $\sigma$ we used in Lemma~\ref{lem:coreset1}, however, notice that in all cases both $t$ and $\sigma'(t)$ belong to the same cell defined by the exponential grids as constructed and processed by the algorithm.
    In fact, both $t$ and $\sigma'(t)$ belong in the same cell $\square_{i'(t)}\in \widebar{V}_{i'(t)}$ around a center $x_{i'(t)}\in \kapprox$. By construction, condition~\eqref{eq:cond} is satisfied for $x_{i'(t)}$ and $\square_{i'(t)}$.
    Hence, the same properties hold. We can distinguish between $\dist(t,x_{i'(t)})\leq \Phi$ and $\dist(t,x_{i'(t)})>\Phi$ as we did in Lemma~\ref{lem:coreset1} making the same arguments. Using the proof of Lemma~\ref{lem:coreset1}, we get that
    $$\eps' \sum_{t\in \Q_u(\I)}\dist(t,\sigma'(t))\leq (\eps')^2\kmedian_{Y}(\Q_u(\I))\leq \eps'\kmedian_{Y}(\Q_u(\I)).$$
    The first inequality $\kmedian_Y(\bar{P}_u)\leq \eps'\kmedian_Y(P_u)+\eps'\kmedian_{Y}(\Q_u(\I))$ follows.

\end{proof}

\begin{proof}[Proof of Lemma~\ref{lem:main2}]

 Next, we assume that  $\DkmedianAlg_\gamma$ is used. We have,
 \begin{align*}
        \kmedian_{\ret}(\coreset)&\leq \gamma\kmedian_{\optD(\coreset)}(\coreset)\leq 2\gamma\kmedian_{\opt(\coreset)}(\coreset)\leq 2\gamma\kmedian_{\opt(P_u)}(\coreset)\leq 2(1+9\eps')\gamma\kmedian_{\opt(P_u)}(P_u)\\&\leq 2(1+9\eps')\gamma\kmedian_{\opt(\Q_u(\I))}(\Q_u(\I)).
    \end{align*}
    Hence,
    \begin{align*}
       \kmedian_{\ret}(\Q_u(\I))&\leq (1+4\eps')\kmedian_{\ret}(P_u)\leq \frac{1+4\eps'}{1-9\eps'}\kmedian_{\ret}(\coreset)=r_u\leq 2\frac{(1+4\eps')(1+9\eps')}{1-9\eps'}\gamma\kmedian_{\opt(\Q_u(\I))}(\Q_u(\I))\\
       &\leq 2(1+34\eps')\gamma\kmedian_{\opt(\Q_u(\I))}(\Q_u(\I))=2(1+\eps)\gamma\kmedian_{\opt(\Q_u(\I))}(\Q_u(\I))\\&\leq 2(1+\eps)\gamma\kmedian_{\optD(\Q_u(\I))}(\Q_u(\I)).
   \end{align*}
   The result follows setting $\eps\leftarrow \eps/2$.
\end{proof}

\section{Missing proofs from Subsection~\ref{subsec:algmedian}}
\label{appndx:finalAlg}
\begin{proof}[Proof of Lemma~\ref{lem:finalApprox}]
    We first consider that $\kmedianAlg_\gamma$ is used.
    Let $\mathcal{O}_v=\pi_{\allattr_v}(\opt(\all))$, and $\mathcal{O}_z=\pi_{\allattr_z}(\opt(\all))$. We define $\mathcal{O}=\mathcal{O}_v\times \mathcal{O}_z$. Notice that $\opt(\all)\subseteq \mathcal{O}$ so $\kmedian_{\mathcal{O}}(\all)\leq \kmedian_{\opt(\all)}(\all)$.
    We have,
    \begin{align*}
        \kmedian_{\kapprox}&(\all)=\sum_{t\in \Q(\I)}||\pi_{\allattr_\node}(t)-\nn(\kapprox,\pi_{\allattr_\node}(t))||\\
        &=\sum_{t\in\Q(\I)}\sqrt{||\pi_{\allattr_v}(t)-\pi_{\allattr_v}(\nn(\kapprox,\pi_{\allattr_\node}(t)))||^2+||\pi_{\allattr_z}(t)-\pi_{\allattr_z}(\nn(\kapprox,\pi_{\allattr_\node}(t)))||^2}\\
        &\leq \sum_{t\in\Q(\I)}||\pi_{\allattr_v}(t)-\pi_{\allattr_v}(\nn(\kapprox,\pi_{\allattr_\node}(t)))||+\sum_{t\in\Q(\I)}||\pi_{\allattr_z}(t)-\pi_{\allattr_z}(\nn(\kapprox,\pi_{\allattr_\node}(t)))||\\
        &= \sum_{t\in\Q(\I)}||\pi_{\allattr_v}(t)-\nn(\ret_v,\pi_{\allattr_v}(t))||+\sum_{t\in\Q(\I)}||\pi_{\allattr_z}(t)-\nn(\ret_z,\pi_{\allattr_z}(t))||\\
        &\leq r_v+r_z\\
        &\leq (1+\eps)\gamma\left(\sum_{t\in\Q(\I)}\!\!\!||\pi_{\allattr_v}(t)\!-\!\nn(\opt(\Q_v(\I)),\pi_{\allattr_v}(t))||\!+\!\!\!\!\sum_{t\in\Q(\I)}\!\!\!\!||\pi_{\allattr_z}(t)-\nn(\opt(\Q_z(\I)),\pi_{\allattr_z}(t))||\right)\\
        &\leq (1+\eps)\gamma\sum_{t\in\Q(\I)}||\pi_{\allattr_v}(t)-\nn(\mathcal{O}_v,\pi_{\allattr_v}(t))||+(1+\eps)\gamma\sum_{t\in\Q(\I)}||\pi_{\allattr_z}(t)-\nn(\mathcal{O}_z,\pi_{\allattr_z}(t))||\\
        &= (1+\eps)\gamma\sum_{t\in\Q(\I)}\left(||\pi_{\allattr_v}(t)-\pi_{\allattr_v}(\nn(\mathcal{O},\pi_{\allattr_\node}(t)))||+||\pi_{\allattr_z}(t)-\pi_{\allattr_z}(\nn(\mathcal{O},\pi_{\allattr_\node}(t)))||\right)\\
        &\leq (1+\eps)\gamma\sqrt{2}\sum_{t\in \Q(\I)}\sqrt{||\pi_{\allattr_v}(t)-\pi_{\allattr_v}(\nn(\mathcal{O},\pi_{\allattr_\node}(t)))||^2+||\pi_{\allattr_z}(t)-\pi_{\allattr_z}(\nn(\mathcal{O},\pi_{\allattr_\node}(t)))||^2}\\
        &=(1+\eps)\gamma\sqrt{2}\sum_{t\in\Q(\I)}||\pi_{\allattr_\node}(t)-\nn(\mathcal{O},\pi_{\allattr_\node}(t))||=(1+\eps)\gamma\sqrt{2}\cdot\kmedian_{\mathcal{O}}(\all)\\
        &\leq (1+\eps)\gamma\sqrt{2}\cdot \kmedian_{\opt(\all)}(\all).
    \end{align*}
    The first and second equalities hold by the definition of the Euclidean metric. The third inequality holds because $\sqrt{a+b}\leq \sqrt{a}+\sqrt{b}$, for $a,b>0$.
    
    We show the fourth inequality by proof by contradiction. Notice that for any $\bar{x}\in\kapprox$, $\pi_{\allattr_v}(\bar{x})\in \ret_v$, because $\kapprox=\ret_v\times \ret_z$.
    Let $s\in\ret_v$ be the closest point in $\ret_v$ from $\pi_{\allattr_v}(t)$, and let $x\in \kapprox$ be the closest point in $\kapprox$ from $\pi_{\allattr_\node}(t)$. Without loss of generality, assume that $s$ is the unique nearest neighbor.
     Assume that $\pi_{\allattr_v}(\nn(\kapprox,\pi_{\allattr_\node}(t)))\neq \nn(\ret_v,\pi_{\allattr_v}(t))\Leftrightarrow \pi_{\allattr_v}(x)\neq s$. In fact, assume that $\pi_{\allattr_v}(x)=s'\in \ret_v$ such that $s'\neq s$, and $\pi_{\allattr_z}(x)=\bar{s}\in \ret_z$.
     Let $x'=s\times \bar{s}\in \kapprox$.
     By definition, $\sqrt{\sum_{A_j\in \allattr_v}(\pi_{A_j}(s)-\pi_{A_j}(t))^2}< \sqrt{\sum_{A_j\in \allattr_v}(\pi_{A_j}(s')-\pi_{A_j}(t))^2}$.
    We have,
    \begin{align*}
    ||x-\pi_{\allattr_\node}(t)||&=\sqrt{\sum_{A_j\in \allattr_\node}(\pi_{A_j}(x)-\pi_{A_j}(t))^2}=\sqrt{\sum_{A_j\in \allattr_v}(\pi_{A_j}(x)-\pi_{A_j}(t))^2 + \sum_{A_j\in \allattr_z}(\pi_{A_j}(x)-\pi_{A_j}(t))^2}\\
    &=\sqrt{\sum_{A_j\in \allattr_v}(\pi_{A_j}(s')-\pi_{A_j}(t))^2 + \sum_{A_j\in \allattr_z}(\pi_{A_j}(\bar{s})-\pi_{A_j}(t))^2}\\
    &> \sqrt{\sum_{A_j\in \allattr_v}(\pi_{A_j}(s)-\pi_{A_j}(t))^2 + \sum_{A_j\in \allattr_z}(\pi_{A_j}(\bar{s})-\pi_{A_j}(t))^2}=\sqrt{\sum_{A_j\in \allattr_\node}(\pi_{A_j}(x')-\pi_{A_j}(t))^2}\\
    &=||x'-\pi_{\allattr_\node}(t)||,
    \end{align*}
    which is a contradiction because $x$ is the closest point in $\kapprox$ from $\pi_{\allattr_\node}(t)$.
    
    The fifth inequality holds because $r_v, r_z$ satisfy Equation~\eqref{eq:goal}.
    The sixth inequality holds because of the definition of $\ret_v$ and $\ret_z$, i.e., $\kmedian_{\ret_v}(\Q_v(\I))\leq (1+\eps)\gamma\kmedian_{\opt(\Q_v(\I))}(\Q_v(\I))$ (similarly for $z$).
    The seventh inequality holds because $\kmedian_{\opt(\Q_v(\I))}(\Q_v(\I))\leq \kmedian_{\mathcal{O}_v}(\Q_v(\I))$ (similarly for $z$).
    The eighth equality holds because it is equivalent to the fourth inequality.
    The ninth inequality holds because $a+b\leq \sqrt{2}\sqrt{a^2+b^2}$ for $a,b>0$. The tenth equality holds because of the definition of the Euclidean metric. The eleventh equality holds by definition of $\kmedian_{\mathcal{O}}(\all)$. The last inequality holds because $\opt(\all)\subseteq \mathcal{O}$.
    Overall, $\alpha=(1+\eps)\gamma\sqrt{2}$ which is a constant.
    Furthermore, notice that $r=r_v+r_z$ so the result follows. Moreover, $r_v\leq (1+\eps)\sum_{t\in \Q(\I)}||\pi_{\allattr_v}(t)-\nn(\ret_v,\pi_{\allattr_v}(t))||$ so $r\leq (1+\eps)\sqrt{2}\kmedian_{\kapprox}(\Q_u(\I))$.

    Similarly, we show the analysis assuming $D\kmedianAlg_\gamma$ is used.
    We have,
    \begin{align*}
    &\kmedian_{\kapprox}(\all)\leq \ldots\leq 
    r_v+r_z
        \\&\leq (2+\eps)\gamma\left(\!\sum_{t\in\Q(\I)}\!\!\!\!\!\left(||\pi_{\allattr_v}(t)\!-\!\nn(\optD(\Q_v(\I)),\pi_{\allattr_v}(t))||\!+\!||\pi_{\allattr_z}(t)\!-\!\nn(\optD(\Q_z(\I)),\pi_{\allattr_z}(t))||\right)\right)\\&\leq 2(2+\eps)\gamma\left(\sum_{t\in\Q(\I)}\!\!\!\!\left(||\pi_{\allattr_v}(t)\!-\!\nn(\opt(\Q_v(\I)),\pi_{\allattr_v}(t))||+\!||\pi_{\allattr_z}(t)\!-\!\nn(\opt(\Q_z(\I)),\pi_{\allattr_z}(t))||\right)\!\right)\\&\leq \ldots\leq 2(2+\eps)\gamma\sqrt{2}\kmedian_{\opt(\Q_u(\I))}(\Q_u(\I)).
    \end{align*}    
    The result follows.
    \end{proof}

\section{Extension to cyclic queries}
\label{appndx:generalQueries}
We first start with some definitions.
Let $\Q$ be a query over a set of $d$ attributes $\allattr$ and a set of $m$ relations $\allrel$.

A fractional edge cover of join query $\Q$ is a point $x = \{x_R \mid R\in \allrel\} \in \mathbb{R}^{m}$ such that for any attribute $A \in \allattr$, $\sum_{R \in \allrel_A} x_R \ge 1$, where $\allrel_A$ are all the relations in $\allattr$ that contain the attribute $A$.
As proved in~\cite{atserias2013size}, the maximum output size of a join query $\Q$ is $O(N^{\lVert x \rVert_1})$. Since the above bound holds for any fractional edge cover, we define $\rho = \rho(\Q)$ to be the fractional cover with the smallest $\ell_1$-norm, i.e., $\rho(\Q)$ is the value of the objective function of the optimal solution of linear programming (LP):
\begin{equation}
\label{eq:lp}
    \min \sum_{R \in \allrel} x_R, \;\text{s.t.}\; \forall R \in \allrel: x_R \ge 0 \;\text{and}\; \forall A \in \allattr: \sum_{R \in \allattr_R} x_R \ge 1.
\end{equation}

Next, we give the definition of the Generalized Hypetree Decomposition (GHD).
A GHD of $\Q$ is a pair $(\T, \lambda)$, where
$\T$ is a tree as an ordered set of nodes and $\lambda: \T \to 2^{\allattr}$ is a labeling function which associates to each vertex $u \in \T$ a subset of attributes in $\allattr$, called $\lambda_u$, such that the following conditions are satisfied:
\begin{itemize}[leftmargin=*]
    \item (coverage) For each $R \in \allrel$, there is a node $u \in \T$ such that $\allattr_R\subseteq\lambda_u$, where $\allattr_R$ is the set of attributes contained in $R$;
    \item (connectivity) For each $A \in \allattr$, the set of nodes $\{u \in \T: A \in \lambda_u\}$ forms a connected subtree of $\T$.
\end{itemize}

Given a join query $\Q$, one of its GHD $(\T, \lambda)$ and a node $u \in \T$, the width of $u$ is defined as the optimal fractional edge covering number of its derived hypergraph $(\lambda_u, \E_u)$, where $\E_u = 
\{\allattr_R\cap \lambda_u: R \in \allrel\}$. Given a join query and a GHD $(\T, \lambda)$, the width of $(\T, \lambda)$ is defined as the maximum width over all nodes in $\T$. Then, the fractional hypertree width of a join query follows:
The fractional hypertree width of a join query $\Q$, denoted as $\fhw(\Q)$, is
\[\fhw(\Q) = \min_{(\T, \lambda)} \max_{u \in \T} \rho(\lambda_u, \E_u),\]
i.e., the minimum width over all GHDs.

Overall, $O(N^\fhw)$ is an upper bound on the number of join results materialized for each node in $\T$. It is also the time complexity to compute the join results for each node in $\T$~\cite{atserias2013size}. Hence, we converted our original cyclic query into an acyclic join query (with join tree $\T$) where each relation has $O(N^\fhw)$ tuples. We execute all our algorithms to the new acyclic join query replacing the $N$ factor with the $N^\fhw$ factor in the running time.

\end{document}